\documentclass[review]{elsarticle}
\journal{Journal of \LaTeX\ Templates}

\usepackage[T1]{fontenc}
\usepackage{amsmath,graphicx,amssymb,amsthm,xfrac,dsfont,mathrsfs}
\usepackage{bm}
\usepackage{constants_u}
\usepackage{natbib}

\usepackage{graphicx}
\usepackage{tikz}
\usetikzlibrary{fit}
\usetikzlibrary{arrows.meta}
\usetikzlibrary{bending}
\usetikzlibrary{shapes}
\tikzset{myarrow/.style={arrows={-Latex[round,bend]}}}
\usepackage{circuitikz}
\usepackage{pgfplots}
\pgfplotsset{compat=1.16}
\usepgfplotslibrary{fillbetween}
\usepackage{marginnote}
\usepgfplotslibrary{colorbrewer}
\definecolor{plotclblue}{RGB}{57,106,180}
\colorlet{plotclblue}{black!20!plotclblue}
\definecolor{plotclorange}{RGB}{218,124,48}
\colorlet{plotclorange}{black!10!plotclorange}
\definecolor{plotclgreen}{RGB}{62,150,81}
\colorlet{plotclgreen}{black!10!plotclgreen}
\definecolor{plotclred}{RGB}{204,37,41}
\colorlet{plotclred}{black!10!plotclred}
\definecolor{plotclblack}{RGB}{83,81,84}
\colorlet{plotclblack}{black!10!plotclblack}
\definecolor{plotclviolet}{RGB}{107,76,154}
\colorlet{plotclviolet}{black!10!plotclviolet}
\definecolor{plotclredbrown}{RGB}{146,36,40}
\colorlet{plotclredbrown}{black!10!plotclredbrown}
\definecolor{plotclocer}{RGB}{148,139,61}
\colorlet{plotclocer}{black!10!plotclocer}

\newtheorem{theorem}{Theorem}
\newtheorem*{theoremnonumber}{Theorem}
\newtheorem{corollary}{Corollary}
\newtheorem{lemma}{Lemma}

\theoremstyle{definition}
\newtheorem{definition}{Definition}

\theoremstyle{remark}
\newtheorem{remark}{Remark}

\renewcommand{\Re}[0]{\ensuremath{\mathop{\mathrm{Re}}}}
\renewcommand{\Im}[0]{\ensuremath{\mathop{\mathrm{Im}}}}
\def\e{\mathop{\mathrm{e}}\nolimits}
\newcommand{\abstrSet}{\mathcal{U}}

\newcommand{\bbset}[1]{\mathbb{#1}}
\newcommand{\N}[0]{\bbset{N}}      
\newcommand{\Z}[0]{\bbset{Z}}      
\newcommand{\R}[0]{\bbset{R}}      
\newcommand{\C}[0]{\bbset{C}}      
\newcommand{\Rc}[0]{\bbset{R}_c}   
\newcommand{\Cc}[0]{\bbset{C}_c}      
\newcommand{\Cell}[0]{\mathcal{C}\ell}  
\newcommand{\Cczero}[0]{\mathcal{C}\ell^{\infty}}  
\newcommand{\Bs}[0]{\mathcal{B}}      
\newcommand{\CBs}[0]{\mathcal{C}\mathcal{B}}      

\newcommand{\iu}[0]{i}                
\newcommand{\di}[1]{\;\mathrm{d}#1}
\newcommand{\spacedot}[0]{\,\cdot\,}      
\DeclareMathOperator{\sinc}{sinc}            
\DeclareMathOperator*{\esssup}{ess\,sup}

\newcommand{\Aset}[0]{\mathcal{A}}

\newcommand{\rR}[0]{|_{\R}}  

\newcommand{\TM}[0]{\text{TM}}

\newcommand{\sq}[1]{(#1)}

\newcommand{\stdR}{\mathscr{R}}
\newcommand{\stdC}{\mathscr{C}}
\newcommand{\img}{\mathrm{img}}
\newcommand{\SOp}[1]{S_{#1}}
\newcommand{\TOp}[1]{T_{#1}}
\newcommand{\TMSOp}[1]{\TM_{S}^{#1}}
\newcommand{\TMTOp}[1]{\TM_{T}^{#1}}

	\tikzset{sr-ff/.style={flipflop, flipflop def={
		t1=S, t2=CP, t3=R, t4={\ctikztextnot{Q}},
		t6=Q, nd=1}},
		}

\newcounter{x}
\newcounter{y}
\newcounter{z}

\newcommand\xaxis{210}
\newcommand\yaxis{-30}
\newcommand\zaxis{90}

\newcommand\topsideG[3]{
  \fill[fill=black!60!green, draw=black,shift={(\xaxis:#1)},shift={(\yaxis:#2)},
  shift={(\zaxis:#3)}] (0,0) -- (30:1) -- (0,1) --(150:1)--(0,0);
}

\newcommand\leftsideG[3]{
  \fill[fill=black!40!green, draw=black,shift={(\xaxis:#1)},shift={(\yaxis:#2)},
  shift={(\zaxis:#3)}] (0,0) -- (0,-1) -- (210:1) --(150:1)--(0,0);
}

\newcommand\rightsideG[3]{
  \fill[fill=black!20!green, draw=black,shift={(\xaxis:#1)},shift={(\yaxis:#2)},
  shift={(\zaxis:#3)}] (0,0) -- (30:1) -- (-30:1) --(0,-1)--(0,0);
}

\newcommand\cubeG[3]{
  \topsideG{#1}{#2}{#3} \leftsideG{#1}{#2}{#3} \rightsideG{#1}{#2}{#3}
}

\newcommand\planepartitionG[1]{
 \setcounter{x}{-1}
  \foreach \a in {#1} {
    \addtocounter{x}{1}
    \setcounter{y}{-1}
    \foreach \b in \a {
      \addtocounter{y}{1}
      \setcounter{z}{-1}
      \foreach \c in {1,...,\b} {
        \addtocounter{z}{1}
        \cubeG{\value{x}}{\value{y}}{\value{z}}
      }
    }
  }
}

\newcommand\topsideR[3]{
  \fill[fill=black!60!red, draw=black,shift={(\xaxis:#1)},shift={(\yaxis:#2)},
  shift={(\zaxis:#3)}] (0,0) -- (30:1) -- (0,1) --(150:1)--(0,0);
}

\newcommand\leftsideR[3]{
  \fill[fill=black!40!red, draw=black,shift={(\xaxis:#1)},shift={(\yaxis:#2)},
  shift={(\zaxis:#3)}] (0,0) -- (0,-1) -- (210:1) --(150:1)--(0,0);
}

\newcommand\rightsideR[3]{
  \fill[fill=black!20!red, draw=black,shift={(\xaxis:#1)},shift={(\yaxis:#2)},
  shift={(\zaxis:#3)}] (0,0) -- (30:1) -- (-30:1) --(0,-1)--(0,0);
}

\newcommand\cubeR[3]{
  \topsideR{#1}{#2}{#3} \leftsideR{#1}{#2}{#3} \rightsideR{#1}{#2}{#3}
}

\newcommand\planepartitionR[1]{
 \setcounter{x}{-1}
  \foreach \a in {#1} {
    \addtocounter{x}{1}
    \setcounter{y}{-1}
    \foreach \b in \a {
      \addtocounter{y}{1}
      \setcounter{z}{-1}
      \foreach \c in {1,...,\b} {
        \addtocounter{z}{1}
        \cubeR{\value{x}}{\value{y}}{\value{z}}
      }
    }
  }
}

\newcommand\topsideB[3]{
  \fill[fill=black!60!blue, draw=black,shift={(\xaxis:#1)},shift={(\yaxis:#2)},
  shift={(\zaxis:#3)}] (0,0) -- (30:1) -- (0,1) --(150:1)--(0,0);
}

\newcommand\leftsideB[3]{
  \fill[fill=black!40!blue, draw=black,shift={(\xaxis:#1)},shift={(\yaxis:#2)},
  shift={(\zaxis:#3)}] (0,0) -- (0,-1) -- (210:1) --(150:1)--(0,0);
}

\newcommand\rightsideB[3]{
  \fill[fill=black!20!blue, draw=black,shift={(\xaxis:#1)},shift={(\yaxis:#2)},
  shift={(\zaxis:#3)}] (0,0) -- (30:1) -- (-30:1) --(0,-1)--(0,0);
}

\newcommand\cubeB[3]{
  \topsideB{#1}{#2}{#3} \leftsideB{#1}{#2}{#3} \rightsideB{#1}{#2}{#3}
}

\newcommand\planepartitionB[1]{
 \setcounter{x}{-1}
  \foreach \a in {#1} {
    \addtocounter{x}{1}
    \setcounter{y}{-1}
    \foreach \b in \a {
      \addtocounter{y}{1}
      \setcounter{z}{-1}
      \foreach \c in {1,...,\b} {
        \addtocounter{z}{1}
        \cubeB{\value{x}}{\value{y}}{\value{z}}
      }
    }
  }
}

\newcommand\topsideGO[3]{
  \fill[fill=black!60!green!35!white, draw=black,shift={(\xaxis:#1)},shift={(\yaxis:#2)},
  shift={(\zaxis:#3)}] (0,0) -- (30:1) -- (0,1) --(150:1)--(0,0);
}

\newcommand\leftsideGO[3]{
  \fill[fill=black!40!green!35!white, draw=black,shift={(\xaxis:#1)},shift={(\yaxis:#2)},
  shift={(\zaxis:#3)}] (0,0) -- (0,-1) -- (210:1) --(150:1)--(0,0);
}

\newcommand\rightsideGO[3]{
  \fill[fill=black!20!green!35!white, draw=black,shift={(\xaxis:#1)},shift={(\yaxis:#2)},
  shift={(\zaxis:#3)}] (0,0) -- (30:1) -- (-30:1) --(0,-1)--(0,0);
}

\newcommand\cubeGO[3]{
  \topsideGO{#1}{#2}{#3} \leftsideGO{#1}{#2}{#3} \rightsideGO{#1}{#2}{#3}
}

\newcommand\planepartitionGO[1]{
 \setcounter{x}{-1}
  \foreach \a in {#1} {
    \addtocounter{x}{1}
    \setcounter{y}{-1}
    \foreach \b in \a {
      \addtocounter{y}{1}
      \setcounter{z}{-1}
      \foreach \c in {1,...,\b} {
        \addtocounter{z}{1}
        \cubeGO{\value{x}}{\value{y}}{\value{z}}
      }
    }
  }
}

\newcommand\topsideRO[3]{
  \fill[fill=black!60!red!35!white, draw=black,shift={(\xaxis:#1)},shift={(\yaxis:#2)},
  shift={(\zaxis:#3)}] (0,0) -- (30:1) -- (0,1) --(150:1)--(0,0);
}

\newcommand\leftsideRO[3]{
  \fill[fill=black!40!red!35!white, draw=black,shift={(\xaxis:#1)},shift={(\yaxis:#2)},
  shift={(\zaxis:#3)}] (0,0) -- (0,-1) -- (210:1) --(150:1)--(0,0);
}

\newcommand\rightsideRO[3]{
  \fill[fill=black!20!red!35!white, draw=black,shift={(\xaxis:#1)},shift={(\yaxis:#2)},
  shift={(\zaxis:#3)}] (0,0) -- (30:1) -- (-30:1) --(0,-1)--(0,0);
}

\newcommand\cubeRO[3]{
  \topsideRO{#1}{#2}{#3} \leftsideRO{#1}{#2}{#3} \rightsideRO{#1}{#2}{#3}
}

\newcommand\planepartitionRO[1]{
 \setcounter{x}{-1}
  \foreach \a in {#1} {
    \addtocounter{x}{1}
    \setcounter{y}{-1}
    \foreach \b in \a {
      \addtocounter{y}{1}
      \setcounter{z}{-1}
      \foreach \c in {1,...,\b} {
        \addtocounter{z}{1}
        \cubeRO{\value{x}}{\value{y}}{\value{z}}
      }
    }
  }
}

\newcommand\topsideBO[3]{
  \fill[fill=black!60!blue!35!white, draw=black,shift={(\xaxis:#1)},shift={(\yaxis:#2)},
  shift={(\zaxis:#3)}] (0,0) -- (30:1) -- (0,1) --(150:1)--(0,0);
}

\newcommand\leftsideBO[3]{
  \fill[fill=black!40!blue!35!white, draw=black,shift={(\xaxis:#1)},shift={(\yaxis:#2)},
  shift={(\zaxis:#3)}] (0,0) -- (0,-1) -- (210:1) --(150:1)--(0,0);
}

\newcommand\rightsideBO[3]{
  \fill[fill=black!20!blue!35!white, draw=black,shift={(\xaxis:#1)},shift={(\yaxis:#2)},
  shift={(\zaxis:#3)}] (0,0) -- (30:1) -- (-30:1) --(0,-1)--(0,0);
}

\newcommand\cubeBO[3]{
  \topsideBO{#1}{#2}{#3} \leftsideBO{#1}{#2}{#3} \rightsideBO{#1}{#2}{#3}
}

\newcommand\planepartitionBO[1]{
 \setcounter{x}{-1}
  \foreach \a in {#1} {
    \addtocounter{x}{1}
    \setcounter{y}{-1}
    \foreach \b in \a {
      \addtocounter{y}{1}
      \setcounter{z}{-1}
      \foreach \c in {1,...,\b} {
        \addtocounter{z}{1}
        \cubeBO{\value{x}}{\value{y}}{\value{z}}
      }
    }
  }
}

\newcommand\topsideBWO[3]{
  \fill[fill=black!60!white!35!white, draw=black,shift={(\xaxis:#1)},shift={(\yaxis:#2)},
  shift={(\zaxis:#3)}] (0,0) -- (30:1) -- (0,1) --(150:1)--(0,0);
}

\newcommand\leftsideBWO[3]{
  \fill[fill=black!40!white!35!white, draw=black,shift={(\xaxis:#1)},shift={(\yaxis:#2)},
  shift={(\zaxis:#3)}] (0,0) -- (0,-1) -- (210:1) --(150:1)--(0,0);
}

\newcommand\rightsideBWO[3]{
  \fill[fill=black!20!white!35!white, draw=black,shift={(\xaxis:#1)},shift={(\yaxis:#2)},
  shift={(\zaxis:#3)}] (0,0) -- (30:1) -- (-30:1) --(0,-1)--(0,0);
}

\newcommand\cubeBWO[3]{
  \topsideBWO{#1}{#2}{#3} \leftsideBWO{#1}{#2}{#3} \rightsideBWO{#1}{#2}{#3}
}

\newcommand\planepartitionBWO[1]{
 \setcounter{x}{-1}
  \foreach \a in {#1} {
    \addtocounter{x}{1}
    \setcounter{y}{-1}
    \foreach \b in \a {
      \addtocounter{y}{1}
      \setcounter{z}{-1}
      \foreach \c in {1,...,\b} {
        \addtocounter{z}{1}
        \cubeBWO{\value{x}}{\value{y}}{\value{z}}
      }
    }
  }
}

	\newcommand{\PlotContinuous}{	\begin{tikzpicture}	\begin{axis}	[	declare function=	{	sinchalf(\x) = sin(deg(pi * \x / 3)) / (pi * \x / 3);
																																											sinc(\x,\o) = sin(deg(pi * (\x - \o))) / (pi * (\x - \o));
																																										},
																																	height=0.7\linewidth,
																																	width=1.05\linewidth,
																																	xmin=-9,
																																	xmax=9,
																																	ymin=-0.25,
																																	ymax=1.1,
																																	axis line style={shorten >=-10pt, shorten <=-10pt},
																																	axis y line=center,
																																	axis x line=middle,
																																	ticks = none
																																]
																																\addplot	[name path=plot1,very thick,plotclred,domain=-9:9,samples=204]
																																					{{	sinchalf(x)}};
																																\addplot	[semithick, densely dotted,plotclblue,domain=-9:9,samples=204]
																																					{{	sinc(x,0)}};
																																\addplot	[semithick, densely dotted,plotclblue,domain=-9:9,samples=204]
																																					{{	sinchalf(1) * sinc(x,1)}};
																																\addplot	[semithick, densely dotted,plotclblue,domain=-9:9,samples=204]
																																					{{	sinchalf(2) * sinc(x,2)}};
																																\addplot	[semithick, densely dotted,plotclblue,domain=-9:9,samples=204]
																																					{{	sinchalf(-1) * sinc(x,-1)}};
																																\addplot	[semithick, densely dotted,plotclblue,domain=-9:9,samples=204]
																																					{{	sinchalf(-2) * sinc(x,-2)}};
																																\addplot	[name path=plot2,very thick, densely dashed, plotclgreen,domain=-9:9,samples=204]
																																					{{%
																																							sinchalf(-3) * sinc(x,-3) +
																																							sinchalf(-2) * sinc(x,-2) +
																																							sinchalf(-1) * sinc(x,-1) +
																																							sinc(x,0) +
																																							sinchalf(1) * sinc(x,1) +
																																							sinchalf(2) * sinc(x,2) +
																																							sinchalf(3) * sinc(x,3)
																																					}};
																																\addplot	[black!10!white] 
																																					fill between[of=plot1 and plot2];
																									\end{axis}
															\end{tikzpicture}
}

\newcommand{\PlotDiscrete}{		\begin{tikzpicture}								\begin{axis}	[	declare function=	{	sinchalf(\x) = sin(deg(pi * \x / 3)) / (pi * \x / 3);
																																											sinc(\x,\o) = sin(deg(pi * (\x - \o))) / (pi * (\x - \o));
																																										},
																																	height=0.7\linewidth,
																																	width=1.05\linewidth,
																																	xmin=-9,
																																	xmax=9,
																																	ymin=-0.25,
																																	ymax=1.1,
																																	axis line style={shorten >=-10pt, shorten <=-10pt},
																																	axis y line=center,
																																	axis x line=middle,
																																	ticks = none
																																]
																																\fill[color = black!15!white] (-9,1.1) rectangle (9,-0.2);
																																\draw[color = white, very thick] (-9,1.1) -- (-9,-0.2);
																																\draw[color = white, very thick] (-8,1.1) -- (-8,-0.2);
																																\draw[color = white, very thick] (-7,1.1) -- (-7,-0.2);
																																\draw[color = white, very thick] (-6,1.1) -- (-6,-0.2);
																																\draw[color = white, very thick] (-5,1.1) -- (-5,-0.2);
																																\draw[color = white, very thick] (-4,1.1) -- (-4,-0.2);
																																\draw[color = white, very thick] (-3,1.1) -- (-3,-0.2);
																																\draw[color = white, very thick] (-2,1.1) -- (-2,-0.2);
																																\draw[color = white, very thick] (-1,1.1) -- (-1,-0.2);
																																\draw[color = white, very thick] (0,1.1) -- (0,-0.2);
																																\draw[color = white, very thick] (1,1.1) -- (1,-0.2);
																																\draw[color = white, very thick] (2,1.1) -- (2,-0.2);
																																\draw[color = white, very thick] (3,1.1) -- (3,-0.2);
																																\draw[color = white, very thick] (4,1.1) -- (4,-0.2);
																																\draw[color = white, very thick] (5,1.1) -- (5,-0.2);
																																\draw[color = white, very thick] (6,1.1) -- (6,-0.2);
																																\draw[color = white, very thick] (7,1.1) -- (7,-0.2);
																																\draw[color = white, very thick] (8,1.1) -- (8,-0.2);
																																\draw[color = white, very thick] (9,1.1) -- (9,-0.2);
																																\draw[color = black!35!white] (-8.5, 0.2) node {?};
																																\draw[color = black!35!white] (-7.5, 0.2) node {?};
																																\draw[color = black!35!white] (-6.5, 0.2) node {?};
																																\draw[color = black!35!white] (-5.5, 0.2) node {?};
																																\draw[color = black!35!white] (-4.5, 0.2) node {?};
																																\draw[color = black!35!white] (-3.5, 0.2) node {?};
																																\draw[color = black!35!white] (-2.5, 0.2) node {?};
																																\draw[color = black!35!white] (-1.5, 0.2) node {?};
																																\draw[color = black!35!white] (-0.5, 0.2) node {?};
																																\draw[color = black!35!white] (0.5, 0.2) node {?};
																																\draw[color = black!35!white] (1.5, 0.2) node {?};
																																\draw[color = black!35!white] (2.5, 0.2) node {?};
																																\draw[color = black!35!white] (3.5, 0.2) node {?};
																																\draw[color = black!35!white] (4.5, 0.2) node {?};
																																\draw[color = black!35!white] (5.5, 0.2) node {?};
																																\draw[color = black!35!white] (6.5, 0.2) node {?};
																																\draw[color = black!35!white] (7.5, 0.2) node {?};
																																\draw[color = black!35!white] (8.5, 0.2) node {?};
																																\draw (-9,0) -- (9,0);
																																\draw (0,-0.2) -- (0,1.1);																
																																\addplot	[ycomb, plotclgreen, very thick ,domain=-9:9]
																																					coordinates
																																					{		(-9,0)(-8,0)(-7,0)(-6,0)(-5,0)(-4,0)
																																							(-3, {sinchalf(-3)})%
																																							(-2, {sinchalf(-2)})%
																																							(-1, {sinchalf(-1)})%
																																							(0, 1)%
																																							(1, {sinchalf(1)})%
																																							(2, {sinchalf(2)})%
																																							(3, {sinchalf(3)})%
																																							(4,0)(5,0)(6,0)(7,0)(8,0)(9,0)
																																					};
																																\addplot	[mark = x, only marks, plotclgreen, very thick, mark size = 3pt, domain=-9:9]
																																					coordinates
																																					{		(-9,0)(-8,0)(-7,0)(-6,0)(-5,0)(-4,0)
																																							(-3, {sinchalf(-3)})%
																																							(-2, {sinchalf(-2)})%
																																							(-1, {sinchalf(-1)})%
																																							(0, 1)%
																																							(1, {sinchalf(1)})%
																																							(2, {sinchalf(2)})%
																																							(3, {sinchalf(3)})%
																																							(4,0)(5,0)(6,0)(7,0)(8,0)(9,0)
																																					};
																									\end{axis}
															\end{tikzpicture}
}

	\vfuzz \hfuzz
\begin{document}

\begin{frontmatter}
\title{On the Need of Analog Signals and Systems for Digital-Twin Representations}

\author{Holger Boche\fnref{fna}}
\fntext[fna]{Holger Boche is with the Technische Universit\"at M\"unchen,
							Lehrstuhl f\"ur Theoretische Informationstechnik, 80290 Munich, Germany, and the Munich
							Center for Quantum Science and Technology (MCQST), Schellingstr. 4, 80799 
							Munich, Germany. e-mail: boche@tum.de.}
\author{Yannik N. Böck\fnref{fnb}}
\fntext[fnb]{Yannik~N.~B{\"o}ck is with the Technische Universit\"at M\"unchen,
							Lehrstuhl f\"ur Theoretische Informationstechnik, 80290 Munich, Germany. e-mail: yannik.boeck@tum.de.}
\author{Ullrich J. Mönich\fnref{fnc}}
\fntext[fnc]{Ullrich~J.~M{\"o}nich is with the Technische Universit\"at M\"unchen,
							Lehrstuhl f\"ur Theoretische Informationstechnik, 80290 Munich, Germany. e-mail: moenich@tum.de.} 
\author{Frank~H.~P.~Fitzek\fnref{fnd}}
\fntext[fnd]{F.~H.~P.~Fitzek is with the Deutsche Telekom Chair of Communication
										Networks, Technical University of Dresden, 01187 Dresden, Germany, and the
										Cluster of Excellence “Centre for Tactile Internet with Human-in-the-Loop”
										(CeTI. email: frank.fitzek@tu-dresden.de}     
\begin{abstract} 
	\normalsize
	We consider the task of converting different digital descriptions of analog bandlimited signals and systems into 
	each other, with a rigorous application of mathematical computability theory. Albeit very fundamental,
	the problem appears in the scope of \emph{digital twinning}, an emerging
	concept in the field of digital processing of analog information that is regularly mentioned 
	as one of the key enablers for next-generation cyber-physical systems and their areas of application. In this context,
	we prove that essential quantities such as the peak-to-average power ratio and the bounded-input/bounded-output norm, which 
	determine the behavior of the real-world analog system, cannot generally be determined 
	from the system's digital twin, depending on which of the above-mentioned descriptions is chosen.
	As a main result, we characterize the algorithmic strength of Shannon's sampling type
	representation as digital twin implementation and also introduce a new digital twin implementation of analog signals and systems.
	We show there exist two digital descriptions, both of which uniquely characterize a certain analog system, such that one description
	can be algorithmically converted into the other, but not vice versa.
\end{abstract}
\end{frontmatter}\pagebreak

\section{Introduction}\label{sec:introduction}

Bandlimited signals are essential to state-of-the-art information processing, especially at the border between
analog and digital systems. In the physical world, be it in signal processing, control, communication or measurement technology, 
information is usually carried by analog, continuous-time signals. In contrast, in the digital world,
where most of the data processing takes place, information is processed in discrete-time computational cycles. 
According to Shannon's sampling theorem, a \emph{sampling series} can be used to 
uniquely recover a bandlimited continuous-time signal from a discrete-time sequence of samples,
provided that the signal's energy is finite and the samples are taken at least at Nyquist rate~\cite{Sh49}. 

Since the publication of Shannon's seminal article, the development and
investigation of sampling-type representations of analog signals and systems has been an active field of research. 
In order to meet the progressive requirements of applied engineering, 
the relevant theory has been advanced into various directions. Among others, this includes extensions to generalized function spaces \cite{Pf71,CaHa82,Gr85}, 
modified sampling sequences \cite{Wa92}, sampling-type representations of operators \cite{Ha01, PfWa16}, and non-deterministic frameworks \cite{Ga72, Ha01}.  
The collected results form the foundation of modern digital processing of analog information.

One of the most recent concepts in the field of digital processing of analog information, which is regularly referred to as one of the key enablers for next-generation cyber-physical systems \cite{BaCaFo19}, is known as 
\emph{digital twinning}. Originally associated primarily with Industry~4.0~\cite{TaZhLiNe19}, the concept is now also attracting great interest in many other areas of modern technology. For recent examples from networking or medicine 
technologies, see~\cite{HoEA21, LaEA22}. With the introduction of the metaverse, real worlds will be transferred to virtual space. 
Here, information processing will be even more dependent on human multi-modalities (human senses) and its interaction with the digital domain, c.f.~\cite{FiEA21}.
In order to make human senses experienceable, the information will have to be processed in real time. This raises the question of whether 
this processing can be done in the digital world at all, or whether analog approaches might be the solution. 

The expectations towards digital twin technologies are ambitious. In medical research, for example, the idea of implementing digital twins of humans is considered,
that can be employed for medical purposes such as virtual surgery. For this kind of technologies, requirements 
regarding trustworthiness are clearly of special relevance. In general, the number of every-day technologies that potentially affect sensitive human goods, like financial resources, 
private information or health, can be expected to rise significantly with the increasing establishment of digital twinning. The need to follow strict
specifications on privacy, integrity, reliability, safety and alike with regards to these technologies, is manifest. 
Considerations of this kind are essential in view of future robotic systems, medicine applications 
and 6G communication technologies, see e.g.~\cite{FiEA20}. We will discuss this topic in detail in Section~\ref{sec:concl}.

Although an unambiguous and widely accepted definition of the term \emph{digital twin} has not yet been established,
the approach usually exhibits the following abstract characteristics:
\begin{itemize}
	\item[1)] 	The starting point is an abstract set of arbitrary entities from 
				the physical world, usually in the context of some engineering problem. The abstract set is often defined
				implicitly by the problem statement. For example, it may consist of all configurations and properties of a network 
				of interacting autonomous vehicles, or of all possible configurations and properties of the individual parts of a 
				combustion engine. Depending on the specific application, there is a number of object related properties that 
				we want to predict. For example, in the case of combustion engines, this may, be the expected fuel 
				consumption in a certain operating state.
	\item[2)]	The objects in the abstract set are assigned a \emph{description} in some \emph{language} that is
				readable by digital machines, see below. For example, the individual parts of the above-mentioned 
				combustion engine may be characterized by means of 
				the finite element method within some computer-aided-design (CAD) software. The machine-readable description is the object's digital 
				\emph{representation}, i.e., a \emph{digital twin}.
				In real-time systems, the object's digital twin is sequentially updated to match its real-world counterpart.
				This is usually implemented by means of (physical) measurements and analog to digital conversion.
	\item[3)]	The digital twin of the object is used as input for an algorithm, which in turn is supposed to predict 
				one of the 	object-related properties. In real-time systems, the output of the algorithm can be used 
				to control the real system. 
\end{itemize}
Schematically, the concept of a digital twin described above is shown in Figure~\ref{fig:exmpI}.
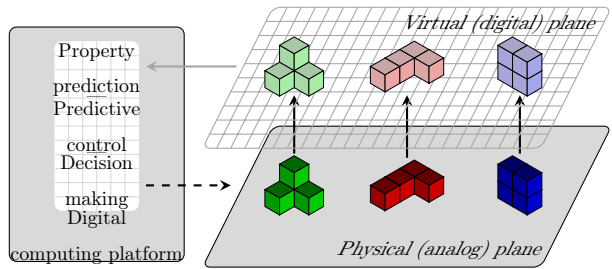
\begin{figure}[!htbp]
	\centering
	\begin{tikzpicture}	[scale=0.75, every node/.style={scale=0.75}]
											\begin{scope}	[	xslant = 0.5,%
																		]
																		\draw[rounded corners = 1mm, fill = black!15!white] (-3, 1.25) rectangle (3, -1.25);
																		\draw (-2.1,0.2) node[scale = 0.3] 	(ItemG)
																																				{	\begin{tikzpicture}	\planepartitionG{{2,1},{1}};
																																					\end{tikzpicture}
																																				};
																		\draw (-0.1,0.2) node[scale = 0.3] 	(ItemR)
																																				{	\begin{tikzpicture}	\planepartitionR{{1,1},{1},{1}};
																																					\end{tikzpicture}
																																				};
																		\draw (1.9,0.2) node[scale = 0.3] 	(ItemB)
																																				{	\begin{tikzpicture}	\planepartitionB{{2,2}};
																																					\end{tikzpicture}
																																				};
																		\draw (2.999, -1.249) node [xslant = 0.5, anchor = south east] {Physical (analog) plane};
																		\draw (-4.85,0.2) node (TMBoxLower) {};
																		\draw (-3.1, 0.2) node (ArrowPointLower) {};
											\end{scope}
											\draw[-stealth, thick] (ItemG) -- +(0,1.6);
											\draw[-stealth, thick] (ItemR) -- +(0,1.6);
											\draw[-stealth, thick] (ItemB) -- +(0,1.6);
											\begin{scope}	[	yshift = 60,%
																			xslant = 0.5
																		]
																		\draw [step = 0.25, color = black!35!white] (-2.999, 1.249) grid (2.999, -1.249);
																		\draw [rounded corners = 1mm, color = black!35!white] (-3, 1.25) rectangle (3, -1.25);
																		\draw (-2.1,0.2) node[scale = 0.3] 	{	\begin{tikzpicture}	\planepartitionGO{{2,1},{1}};
																																					\end{tikzpicture}
																																				};
																		\draw (-0.1,0.2) node[scale = 0.3] 	{	\begin{tikzpicture}	\planepartitionRO{{1,1},{1},{1}};
																																					\end{tikzpicture}
																																				};
																		\draw (1.9,0.2) node[scale = 0.3] 	{	\begin{tikzpicture}	\planepartitionBO{{2,2}};
																																					\end{tikzpicture}
																																				};
																		\draw (2.999, 1.249) node [xslant = 0.5, anchor = north east] {Virtual (digital) plane};
																		\draw (-4.85,0.2) node (TMBoxUpper) {};
																		\draw (-3, 0.2) node (ArrowPointUpper) {};						
											\end{scope}
											\draw [fill = black!15!white, rounded corners = 1mm] 
														([xshift = -2.3cm, yshift = 0.7cm] TMBoxUpper) rectangle ([xshift = 0.8cm, yshift = -1.35cm] TMBoxLower)
														node[pos = .5] (TextTMBoxPlatform) {};
											\draw [-stealth, thick, color = black!35!white] (ArrowPointUpper) -- (TMBoxUpper);
											\draw [-stealth, thick, dashed] 								(TMBoxLower) -- (ArrowPointLower);
											\fill [color = white, rounded corners = 1mm] 
														([xshift = -1.5cm, yshift = 0.45cm] TMBoxUpper) rectangle ([yshift = -0.45cm] TMBoxLower)
														node[pos = .5] (TextTMBoxAlgorithm) {};
											\draw [step = 0.25, color = black!15!white, rounded corners = 1mm] 
														([xshift = -1.49cm, yshift = 0.44cm] TMBoxUpper) grid ([xshift = -0.001cm, yshift = -0.449cm] TMBoxLower);
											\draw ([yshift =  1.0cm] TextTMBoxAlgorithm) node[anchor = center, align = center] {Property\\ prediction};
											\draw ([yshift =  0.5cm] TextTMBoxAlgorithm) node[anchor = center, align = center] {---};
											\draw ([yshift =  0.0cm] TextTMBoxAlgorithm) node[anchor = center, align = center] {Predictive\\ control};
											\draw ([yshift = -0.5cm] TextTMBoxAlgorithm) node[anchor = center, align = center] {---};
											\draw ([yshift = -1.0cm] TextTMBoxAlgorithm) node[anchor = center, align = center] {Decision\\ making};
											\draw ([yshift = -1.65cm] TextTMBoxPlatform) node[anchor = center, align = center] {Digital\\ computing platform};							
	\end{tikzpicture}
	\caption{Schematic representation of the digital twin approach according to the formalization given in Section~\ref{sec:introduction}.}
	\label{fig:exmpI}
\end{figure}

In the scope of this article, the term ``language'' as used above refers to a fixed method for characterizing abstract objects that is accessible to Turing machines,
e.g., the concept of discrete- and continuous-time descriptions of computable bandlimited signals introduced in Section~\ref{sec:Conv}.
It is not to be confused with a formal language according to the strict mathematical definition. However, it is noteworthy that, since the theory of Turing machines
can be equivalently formalized by the theory of formal languages, it is (in principle) possible to formalize our framework in 
a manner such that it does coincide with formal languages in the mathematical sense. In this context, the problem of
converting different digital descriptions of analog objects into each other may be regarded as a \emph{compiler} problem.

Depending on the individual application, the implementation of a digital twin can be arbitrarily complex. 
For example, the source code of a CAD application can be thousands of lines long, and the data describing a particular object can be several gigabytes in size. 
Digital twins of humans in healthcare and robotics can be expected to be even more complex than that.
Thus, the question of the "proper" way to describe a real world object arises: if a certain property about the real system should be predicted, which characteristics does the 
language describing the system has to satisfy?
The problem is visualized in Figure~\ref{fig:exmpII}.
\begin{figure}[!htbp]
	\centering
	\begin{tikzpicture}	
											\newlength{\yshiftUpper};
											\setlength{\yshiftUpper}{60pt};
											\begin{scope}	[	xslant = 0.5,%
																		]
																		\draw[rounded corners = 1mm, fill = black!15!white] (-3, 1.25) rectangle (3, -1.25);
																		\draw (-0.2,0.2) node[] 	(ItemG) {};
																		\draw (2.999, -1.249) node [xslant = 0.5, anchor = south east] {Physical (analog) plane};
											\end{scope}
											\draw (ItemG) node[scale = 0.3] 	{		\begin{tikzpicture}	\planepartitionG{{2,1},{1}};
																														\end{tikzpicture}
																												};
											\begin{scope}	[	yshift = \yshiftUpper,%
																			xslant = 0.5
																		]
																		\draw (-2.,0.2) node (ItemGO) {};
																		\draw (1.8,0.2) node (ItemBWO) {};
											\end{scope}
											\draw[-stealth, thick]	([xshift = -2.5, yshift = 17.5] ItemG.center) -- %
																							([xshift = -2.5, yshift = 0.4\yshiftUpper] ItemG.center) -- %
																							([yshift = -0.6\yshiftUpper] ItemGO.center) -- %
																							([yshift = -0.25\yshiftUpper] ItemGO.center);
											\draw[-stealth, thick]	([xshift =  2.5, yshift = 17.5] ItemG.center) -- %
																							([xshift =  2.5, yshift = 0.4\yshiftUpper] ItemG.center) -- %
																							([yshift = -0.6\yshiftUpper] ItemBWO.center) -- %
																							([yshift = -0.25\yshiftUpper] ItemBWO.center);
											\begin{scope}	[	yshift = \yshiftUpper,%
																			xslant = 0.5
																		]
																		\draw [step = 0.25, color = black!35!white] (-2.999, 1.249) grid (-0.801, -1.249);
																		\draw [step = 0.25, color = black!35!white] (0.799, 1.249) grid (2.999, -1.249);
																		\draw [rounded corners = 1mm, color = black!35!white] (-3, 1.25) rectangle (-0.8, -1.25);
																		\draw [rounded corners = 1mm, color = black!35!white] (0.8, 1.25) rectangle (3, -1.25);
																		\draw (-0.801, 1.249) node [align = right, xslant = 0.5, anchor = north east] {Language A};	
																		\draw (2.999, 1.249)  node [align = right, xslant = 0.5, anchor = north east] {Language B};
																		\draw (0,0.2)  				node [ellipse, fill = black!15!white, align = center, anchor = center, scale = 0.8] {A or B?};
											\end{scope}
											\draw (ItemGO) 	node[scale = 0.3] 	{		\begin{tikzpicture}	\planepartitionGO{{2,1},{1}};
																															\end{tikzpicture}
																													};
											\draw (ItemBWO) node[scale = 0.3] 	{		\begin{tikzpicture}	\planepartitionBWO{{2,1},{1}};
																															\end{tikzpicture}
																													};									
	\end{tikzpicture}
	\caption{	Digital twins of the same abstract object in different machine-readable languages. Even if both twins uniquely characterize the real object, 
						not all information about it may be algorithmically accessible in both languages. The choice of which language to use is thus a creative engineering
						task and depends strongly on the specific application.}
	\label{fig:exmpII}
\end{figure}
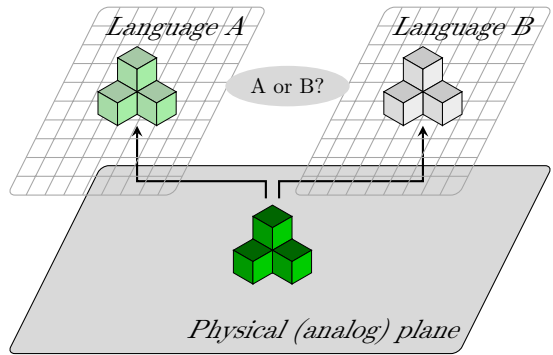
The answer to this question highly depends on the specific application and the properties to be predicted. 
Often, the choice of language is also a matter of feasibility and convenience.

In this article, we consider the standard digital signal processing description of bandlimited signals and systems in connection 
with digital twin technology. Albeit digital twins are mostly associated to complex systems like networks of autonomous vehicles or combustion engines,
the digital processing of bandlimited signals satisfies all above criteria, as will be discussed in Section~\ref{sec:MotivDigitTwin}.
Furthermore, bandlimited signals and systems regularly occur as subsystems of more complex digital twins. For example, MATLAB Simulink relies heavily
on descriptions of LTI systems by means of their discrete impulse response. Since the structure of a digital twin of a bandlimited signal is conceptually simple,
it is an excellent research object for which a mathematically rigorous and well defined 
model of the term digital twin that meets the standards of theoretical informatics can be established.
In order to do so, we employ the theory of Turing machines. Turing machines are among the most refined 
models of digital computers in literature. In fact, the \emph{Church Turing Thesis},
which implicitly states that Turing machines indeed yield a complete characterization of the theoretical capabilities
of digital hardware, is widely accepted in the community of information technology. This allows for the characterization
of fundamental limits of real-world computers: if a certain problem can be proven to be unsolvable on a Turing machine,
it can definitely not be solved on any real-world computer. 

In the context of bandlimited signals and systems, the application of Turing's theory leads to a computable variant of the Bernstein spaces 
$\Bs_{\pi}^{p}$. Generally, these consist of bandlimited signals
with finite $L^p$-norm as characteristic time-domain behavior \cite{higgins99_book}.
The computable Bernstein spaces then consist of those signals whose sequence of sampling values can be generated
algorithmically, along with a bound on the approximation error with respect to the corresponding norm.
More specifically, according to this definition a signal $f \in \Bs_{\pi}^p$ is called
computable if
\begin{enumerate}
\item there exists an algorithm that computes a sequence $\sq{f_n}_{n \in \N}$ of finite
  Shannon sampling series, and
\item the approximation error can be effectively controlled, i.e., we have 
  $\lVert f - f_n \rVert_{\Bs_{\pi}^p} \leq 2^{-M}$ for all $n, M\in \N$ with \(n \geq \xi(M)\)
	and some computable function \(\xi : \N \rightarrow \N\).
\end{enumerate}
Hence, the signal \(f\) is stored in terms of an algorithm that, when executed, produces an approximation of \(f\)
up to arbitrary precision, plus an estimate of the distance to \(f\).
The advantages of the above definition are apparent: the definition is intuitively clear,
very general, and, since it uses the finite Shannon sampling series, it is easy to perform
analytical calculations, such as taking the derivative.

It is important to note that, from an engineering point of view, our approach yields
an ``ideal'' digital twin, in the sense that it characterizes the analog information uniquely.
To illustrate this, consider the simplified case of sampling an analog signal for a finite duration
with a quantized sampling depth. Given a list of sampling values, there exists an uncountably large
set of signals that coincide at the sampling points with the values in the list up to the quantization error.
Hence, the true analog signal cannot be recovered exactly from the list of sampling values.
In our approach, every sampling value can be computed up to arbitrary precision. Hence, 
the digital twin of the analog signal, which, as mentioned above, consists of an algorithm that
produces the sampling values, uniquely characterizes it's analog counterpart.

The remainder of the article is structured as follows. In Section~\ref{sec:MotivDigitTwin},
we will give an in-depth motivation of the general problem in the context of digital twinning.
Sections~\ref{sec:NandBsignals} and~\ref{sec:CBsignals} are dedicated to preliminaries.
In Section~\ref{sec:Conv}, we give a formal characterization of the problem statement 
based on the theory of Turing machines. Then, in Sections~\ref{sec:UncompInfty},~\ref{sec:UncompOne}
and~\ref{sec:CompOfT}, we will establish the main results of our work. In Section~\ref{sec:InterpDT},
we will return to the context of digital twinning and interpret the main results accordingly. The article
closes in Section~\ref{sec:concl} with a brief subsumption of our work.
	
\section{Problem Motivation in the Context of Digital Twins}\label{sec:MotivDigitTwin}
In the most general formalization of the concept of digital twins, the behavior of the physical world
(more precisely: the part of the physical world that is subject to the some engineering problem)
is captured by some general mathematical model, which in turn characterizes an abstract set \(\abstrSet\). 
The set \(\abstrSet\) consists of a selection of objects from the physical world that are to be represented on a digital 
computer. Usually, this set is determined by the engineering problem under consideration.
For example, \(\abstrSet\) may consist of the possible arrangements of the individual parts of some
mechanical device, the points in the phase space of some dynamical system, or a class of analog electronic 
signals. In the digital domain, the counterpart to \(\abstrSet\) is another set \(\mathscr{D}\) that consist
of \emph{machine-readable} descriptions of the objects in \(\abstrSet\), i.e., a machine-readable language. We will discuss a
precise formalization of the term ``machine-readable'' in Section~\ref{sec:CBsignals}. 
Generally, the purpose of a digital twin is to answer questions about the physical 
world by means of a digital computer. The computer is presented with a description \(\mathfrak{D}\in\mathscr{D}\) of some physical
entity \(u\in\abstrSet\) and is supposed to return an answer about some property of the same entity.
Hence, a the digital twin \(\mathfrak{D}\) of \(u\) serves as \emph{input} for an information processing routine.

The language \(\mathscr{D}\) is not unique, nor a priori determined by the set \(\abstrSet\). 
In fact, the choice of the proper language \(\mathscr{D}\) is a creative task that depends on the specific problems that are to be solved by 
means of its elements \(\mathfrak{D}\). This is best illustrated directly by the example of bandlimited signals: 
\begin{itemize}
	\item[1)]	The abstract set \(\CBs_{\pi}^p\) whose elements are to be described in a machine-readable language 
						consists of all (computable) bandlimited signals. In communication technology, for example, the latter may act as an information carrier,
						whereas in control theory, they may characterize an LTI system by means of its impulse response.
						For wireless transmissions, the peak value of the information carrier has to be controlled in order to avoid
						non-linear distortions and inter-band interference. The latter is known as the \emph{PAPR-problem},
						a prominent engineering task that has been investigated in different contexts, 
						see e.g.~\cite{LiYu05,PaTa00,BoTa18}. For a comprehensive overview, we refer to~\cite{WuEA13}.
						In control theory, we may be interested in the \(L_1\) norm of the signal, 
						since it is the crucial quantity in the context of \emph{BIBO stability}, a fundamental 
						topic that is discussed in most of the relevant introductory textbooks.
	\item[2)]	Shannon's sampling theorem allows for the exact characterization of a bandlimited signal in terms of a sequence
						of sampling values. Hence, in digital signal processing, bandlimited signals are usually described
						by storing their sampling values in one way or another. As indicated in Section~\ref{sec:introduction}, we will consider the more general
						case of generating the sampling values algorithmically. In particular, we will introduce the languages
						\(\mathscr{F}^p\) and \(\mathscr{X}^p\), which consist of \emph{continuous-time} and \emph{discrete-time}
						descriptions of bandlimited signals. Their elements \(\mathfrak{F}\in\mathscr{F}^p\), \(\mathfrak{X}\in \mathscr{X}^p\),
						then constitute to algorithms that reproduce the sampling values \(\sq{f(k)}_{k\in\Z}\) of some bandlimited signal \(f\in\CBs_{\pi}^p\).
	\item[3)]	The stored descriptions of \(f\), i.e., \(\mathfrak{F}\) and \(\mathfrak{X}\), respectively, are passed as input to an algorithm that calculates the peak value or 
						BIBO norm of the signal \(f\). If, for example, the BIBO norm exceeds a certain threshold, the system may be throttled in
						order to avoid overshooting at the output. In particular, both quantities should be determined digitally
						by means of \(\mathfrak{F}\) or \(\mathfrak{X}\), before generating the actual signal \(f\) by digital to analog conversion.
\end{itemize}
The two languages used to describe the elements of \(\CBs_{\pi}^p\) are illustrated by Figure~\ref{fig:exmpIV}. 
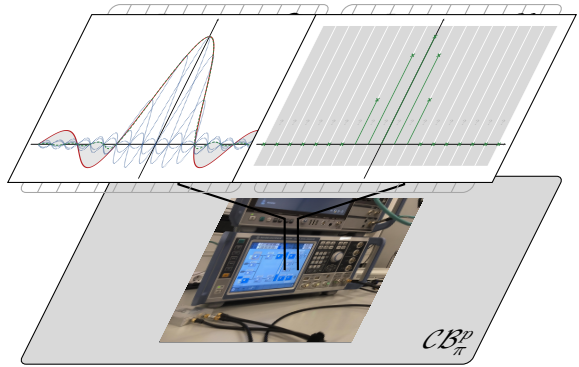
\begin{figure}[!htbp]
	\centering
	\begin{tikzpicture}	
											\setlength{\yshiftUpper}{65pt};
											\begin{scope}	[	xslant = 0.5,%
																		]
																		\draw[rounded corners = 1mm, fill = black!15!white] (-3, 1.25) rectangle (3, -1.25);
																		\draw (-0,0) node[] 	(ItemG) {};
																		\draw (2.999, -1.249) node [xslant = 0.5, anchor = south east] {\(\CBs_{\pi}^{p}\)};
											\end{scope}
											\draw (ItemG) node[draw = black, fill = white, anchor=center,inner sep=0pt, xslant = 0.5] 	{		\includegraphics[width=1in]{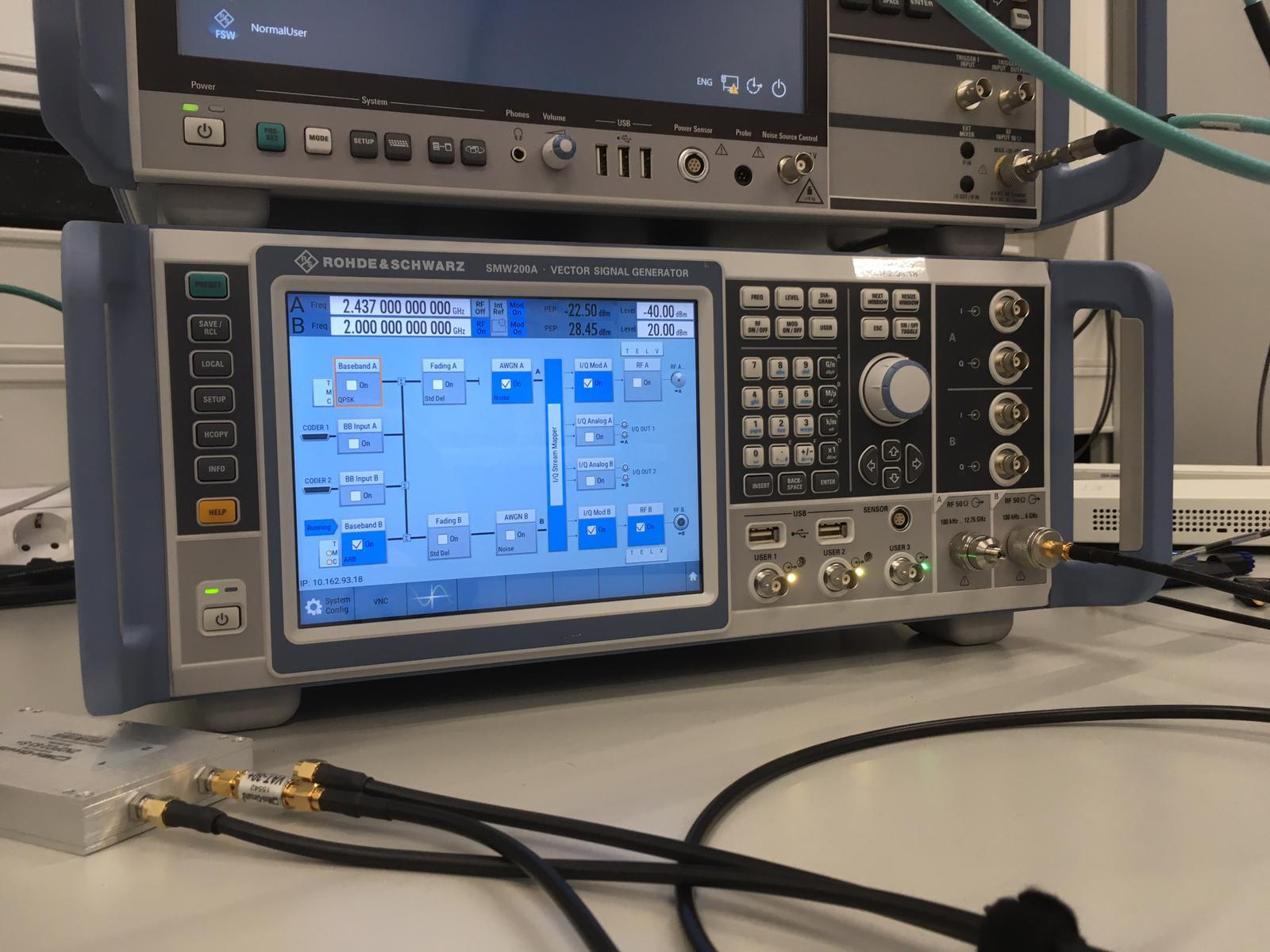}
																																										};																	
											\begin{scope}	[	yshift = \yshiftUpper,%
																			xslant = 0.5
																		]
																		\draw (-1.5,0) node (ItemGO) {};
																		\draw (1.5,0) node (ItemBWO) {};
											\end{scope}
											\draw[-stealth, thick]	([xshift = -2.5, yshift = 0] ItemG.center) -- %
																							([xshift = -2.5, yshift = 0.3\yshiftUpper] ItemG.center) -- %
																							([yshift = -0.5\yshiftUpper] ItemGO.center) -- %
																							([yshift = -0.25\yshiftUpper] ItemGO.center);
											\draw[-stealth, thick]	([xshift =  2.5, yshift = 0] ItemG.center) -- %
																							([xshift =  2.5, yshift = 0.3\yshiftUpper] ItemG.center) -- %
																							([yshift = -0.5\yshiftUpper] ItemBWO.center) -- %
																							([yshift = -0.25\yshiftUpper] ItemBWO.center);
											\begin{scope}	[	yshift = \yshiftUpper,%
																			xslant = 0.5
																		]
																		\draw [step = 0.25, color = black!35!white] (-2.999, 1.249) grid (-0.101, -1.249);
																		\draw [step = 0.25, color = black!35!white] (0.099, 1.249) grid (2.999, -1.249);
																		\draw [rounded corners = 1mm, color = black!35!white] (-3, 1.25) rectangle (-0.1, -1.25);
																		\draw [rounded corners = 1mm, color = black!35!white] (0.1, 1.25) rectangle (3, -1.25);
																		\draw (-0.1, 1.249) node [align = right, xslant = 0.5, anchor = north east] {Language \(\mathscr{F}^p\)};	
																		\draw (3, 1.249)  node [align = right, xslant = 0.5, anchor = north east] {Language \(\mathscr{X}^p\)};
											\end{scope}
											\draw (ItemGO) 	node[draw = black, inner sep = 15pt, anchor = center, fill = white, scale = 0.28, xslant = 0.5] 	{		\PlotContinuous
																																																																			};
											\draw (ItemBWO) node[draw = black, inner sep = 15pt, anchor = center, fill = white, scale = 0.28, xslant = 0.5] 	{		\PlotDiscrete
																																																																			};									
	\end{tikzpicture}
	\caption{Visualization of the "languages" \(\mathscr{F}^p\) and \(\mathscr{X}^p\), both of which describe the set \(\CBs_{\pi}^p\).}
	\label{fig:exmpIV}
\end{figure}
In the scope of this article, we will illustrate that the choice of the proper language, i.e., the exact specification
of how the information describing \(f\) is stored, is crucial in calculating both the peak value and the BIBO norm of \(f\).

We will formally introduce the
sets \(\CBs_{\pi}^p\), \(\mathscr{F}^p\) and \(\mathscr{X}^p\) in Sections~\ref{sec:CBsignals}~and~\ref{sec:Conv}.
Given a signal \(f\in\CBs_{\pi}^{p}\) and corresponding descriptions \(\mathfrak{F}\in \mathscr{F}^p\) and \(\mathfrak{X} \in \mathscr{X}^p\),
one may ask the following questions:
\begin{itemize}
	\item[1)] Can we compute the \(p\)-norm \(\lVert f \rVert_{\Bs_{\pi}^p}\) of \(f\) from the description \(\mathfrak{F}\)?
	\item[2)] Can we compute the \(p\)-norm \(\lVert f \rVert_{\Bs_{\pi}^p}\) of \(f\) from the description \(\mathfrak{X}\)?
\end{itemize}
We will see that both \(\mathfrak{F}\) and \(\mathfrak{X}\) contain all information about the signal \(f\), in the sense that
they characterize the signal \(f\) uniquely. However, as we will prove, for \(p = \infty\), only Question 1 can be answered in the 
positive, while Question 2 has to be answered in the negative. 

As the previous example shows, the question of whether a specific information about \(u\) can be extracted from a digital twin \(\mathfrak{D}\) of
\(u\) may depend crucially on the specific structure of \(\mathscr{D}\). More generally, one may consider the principle of \emph{algorithmic equivalence} of different languages describing the same abstract set, which is depicted in Figure~\ref{fig:exmpIII}. 
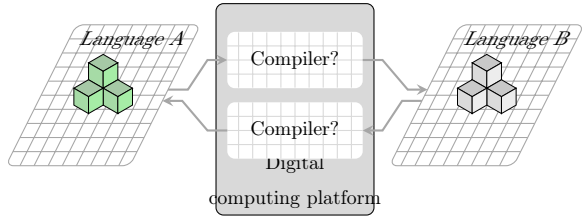
\begin{figure}[!htbp]
	\centering
	\begin{tikzpicture}	[scale=0.75, every node/.style={scale=0.75}]
											\setlength{\yshiftUpper}{-1.25pt};
											\begin{scope}	[	shift = {(0,-0.125)},%
																			xslant = 0.5
																		]
																		\draw (-3.5,0.2) node (ItemGO) {};
																		\draw (3.3,0.2) node (ItemBWO) {};
											\end{scope}
											\begin{scope}	[	shift = {(0,-0.125)},%
																			xslant = 0.5
																		]
																		\draw [step = 0.25, color = black!35!white] (-4.499, 1.249) grid (-2.301, -1.249);
																		\draw [step = 0.25, color = black!35!white] (2.299, 1.249) grid (4.499, -1.249);
																		\draw [rounded corners = 1mm, color = black!35!white] (-4.5, 1.25) rectangle (-2.3, -1.25);
																		\draw [rounded corners = 1mm, color = black!35!white] (2.3, 1.25) rectangle (4.5, -1.25);
																		\draw (-2.301, 1.249) node [align = right, xslant = 0.5, anchor = north east] {Language A};	
																		\draw (4.499, 1.249)  node [align = right, xslant = 0.5, anchor = north east] {Language B};
																		\draw [] (-2.3,0.1) 	node (LangAUpper) {};
																		\draw [] (-2.3,-0.1) 	node (LangALower) {};
																		\draw [] (2.3,0.1) 		node (LangBUpper) {};
																		\draw [] (2.3,-0.1) 	node (LangBLower) {};
											\end{scope}
											\draw (ItemGO) 	node[scale = 0.3] 	{		\begin{tikzpicture}	\planepartitionGO{{2,1},{1}};
																															\end{tikzpicture}
																													};
											\draw (ItemBWO) node[scale = 0.3] 	{		\begin{tikzpicture}	\planepartitionBWO{{2,1},{1}};
																															\end{tikzpicture}
																													};
											\draw [fill = black!15!white, rounded corners = 1mm] (-1.4,1.5) rectangle (1.4,-2.25);
											\draw [align = center, anchor = south] (0,-2.25) node {Digital\\ computing platform};
											\fill [color = white, rounded corners = 1mm] 
														(-1.2,1.0) rectangle (1.2,0.0);
											\draw [] (-1.2, 0.5) 	node (CompilerAUpper) {};
											\draw [] (1.2, 0.5) 	node (CompilerBUpper) {};
											\fill [color = white, rounded corners = 1mm] 
														(-1.2,-0.25) rectangle (1.2,-1.25);
											\draw [] (-1.2, -0.75) 	node (CompilerALower) {};
											\draw [] (1.2, -0.75) 	node (CompilerBLower) {};
											\draw [step = 0.25, color = black!15!white, rounded corners = 1mm] 
														(-1.2,0.99) grid (1.2,0.01);
											\draw [step = 0.25, color = black!15!white, rounded corners = 1mm] 
														(-1.2,-0.251) grid (1.2,-1.249);
											\draw [anchor = center] (0,0.5) node {Compiler?};
											\draw [anchor = center] (0,-0.75) node {Compiler?}; 
											\draw [color = black!35!white, thick, -stealth] 
														(LangAUpper.center) -- 
														([xshift = 10] LangAUpper.center) -- 
														([xshift = -10] CompilerAUpper.center) -- 
														(CompilerAUpper.center);
											\draw [color = black!35!white, thick, -stealth] 
														(CompilerBUpper.center) -- 
														([xshift = 10] CompilerBUpper.center) -- 
														([xshift = -10] LangBUpper.center) -- 
														(LangBUpper.center);
											\draw [color = black!35!white, thick, stealth-] 
														(LangALower.center) -- 
														([xshift = 10] LangALower.center) -- 
														([xshift = -10] CompilerALower.center) -- 
														(CompilerALower.center);
											\draw [color = black!35!white, thick, stealth-] 
														(CompilerBLower.center) -- 
														([xshift = 10] CompilerBLower.center) -- 
														([xshift = -10] LangBLower.center) -- 
														(LangBLower.center);
	\end{tikzpicture}
	\caption{Schematic representation of the principle of \emph{algorithmic equivalence} as formalized in Section~\ref{sec:MotivDigitTwin}.}
	\label{fig:exmpIII}
\end{figure}
Again, let \(\CBs_{\pi}^p\), \(p\in\Rc\), \(1 \leq p \leq \infty\) be the space of computable
bandlimited signals with describing languages \(\mathscr{F}^p\) and \(\mathscr{X}^p\).
Consider the following questions:
\begin{itemize}
	\item[3)] Given a description \(\mathfrak{F} \in \mathscr{F}^p\) of some signal \(f \in \CBs_{\pi}^p\), is it always
		possible to compute a description \(\mathfrak{X} \in \mathscr{X}^p\) of the same signal? 
	\item[4)] Given a description \(\mathfrak{X} \in \mathscr{X}^p\) of some signal \(f \in \CBs_{\pi}^p\), is it always
		possible to compute a description \(\mathfrak{F} \in \mathscr{F}^p\) of the same signal? 
\end{itemize}
If both questions can be answered in the positive, we can consider the languages \(\mathscr{X}^p\) and \(\mathscr{F}^p\) as algorithmically equivalent,
in the sense that all information about \(f \in \CBs_{\pi}^p\) that can be algorithmically extracted from a description
\(\mathfrak{F} \in \mathscr{F}^p\) can also be extracted from a description \(\mathfrak{X} \in \mathscr{X}^p\) and vice versa. However, the negative answer
to question 2 for the case of \(p = \infty\) already shows that \(\mathscr{F}^p\) and \(\mathscr{X}^p\) are not algorithmically equivalent in general.
 
The remaining part of the article deals with a precise analysis of the two languages \(\mathscr{F}^p\) and \(\mathscr{X}^p\) 
describing the abstract set \(\CBs_{\pi}^p\), \(p\in\Rc\), \(1 \leq p \leq \infty\). Special attention is given to the characterization of the 
algorithmic equivalence of \(\mathscr{F}^p\) and \(\mathscr{X}^p\).

\section{Notation and Bandlimited Signals} \label{sec:NandBsignals}

By $\ell^{\infty}$ we denote the set of all complex-valued
sequences indexed by \(\Z\) that vanish at infinity. That is, we have
\begin{align}\label{eq:decay}
	\lim_{k\to\infty} x(k) = \lim_{k\to-\infty} x(k) = 0 
\end{align}
for all \(x = \sq{x(k)}_{k\in\Z}\in\ell^{\infty}\).
Equipped with the \emph{uniform norm}
\begin{align*}
	\lVert x \rVert_{\ell^{\infty}} = \sup_{k \in \Z} \lvert x(k) \rvert,
\end{align*}
the set \(\ell^{\infty}\) becomes a Banach space.
Further, by $\ell^p(\Z)$, $1 \leq p<\infty$, we denote the usual space of $p$th-power
summable sequences with the usual \(p\)-norm
\begin{align*}
	\lVert x \rVert_{\ell^p} = \left(\sum_{k=-\infty}^{\infty} \lvert x(k) \rvert^p\right)^{1/p}.
\end{align*}

For $\Omega \subseteq \R$, let $L^\infty(\Omega)$ be the space of all measurable, complex-valued functions 
on \(\Omega\) for which the essential supremum norm
\begin{align*}
	\lVert f \rVert_{\infty} = \esssup_{t \in \Omega} \lvert f(t) \rvert
\end{align*}
is finite. Then, the space $\Bs_{\sigma}^{\infty}$ consists of all 
entire functions of exponential type at most $\sigma$, whose restriction to the real line
is in $L^{\infty}(\R)$ and vanishes at infinity. Equipped with the essential supremum norm,
the space $\Bs_{\sigma}^{\infty}$ becomes a Banach space. 

Furthermore, let $L^p(\Omega)$, $1\leq p < \infty$, be the space of all
measurable, complex-valued, $p$th-power Lebesgue integrable functions on $\Omega$, with the usual \(p\)-norm
\begin{align*}
	\lVert f \rVert_p = \left( \int_{\Omega} \lvert f(t) \rvert^p \di{t} \right)^{1/p}.
\end{align*}
The \emph{Bernstein space} $\Bs_{\sigma}^p$, $\sigma>0$, $1 \leq p \leq \infty$, consists
of those functions in $\Bs_{\sigma}^{\infty}$ whose restriction to the real line is an element of
\(L^p(\R)\), i.e.,
\begin{align*}
	\Bs_{\sigma}^p := \big\{ f \in \Bs_{\sigma}^{\infty} : f\rR \in L^p(\R) \big\}.
\end{align*}
The norm for $\Bs_{\sigma}^p$ is given by the $L^p$-norm on the real line, i.e.,
$\lVert \spacedot \rVert_{\Bs_{\sigma}^p} = \lVert \spacedot \rVert_{p}$.
For details, we refer to \cite[\emph{Definition}~6.5, p.~49]{higgins96_book}.

\begin{remark}\label{rem:BInfty}
	In the relevant literature, it is common to denote by \(\ell^\infty\) the space of all complex-valued 
	sequences \(\sq{x(k)}_{k\in\Z}\) with bounded uniform norm, without additionally requiring \eqref{eq:decay}
	to hold. The same applies in an analogous way for \(\Bs_{\sigma}^{\pi}\). These spaces are not considered 
	within the scope of this article. Keeping the same notation
	for the restricted spaces of sequences (functions, respectively) that
	additionally vanish asymptotically improves 
	the readability of expressions of the form ``\(x\in\ell^p, 1 \leq p \leq \infty\)'' and
	``\(f\in\Bs_{\sigma}^p, 1 \leq p \leq \infty\)'', see below.
\end{remark}

A signal in $\Bs_{\sigma}^p$ is called bandlimited to $\sigma$.
The set $\Bs_{\sigma}^2$ is the frequently used space of bandlimited signals with bandwidth
$\sigma$ and finite energy, and $\Bs_{\sigma}^{\infty}$ the space of all bandlimited
signals with bandwidth $\sigma$ that are bounded on the real axis and asymptotically vanish at (real) infinity.
We have 
\begin{align}
	\begin{array}	{c c c c c}
								\Bs_{\sigma}^r 	& \subsetneq &\Bs_{\sigma}^s 	& \subsetneq &\Bs_{\sigma}^{\infty} \\
								\ell^r 					& \subsetneq &\ell^s 					& \subsetneq &\ell^{\infty}
	\end{array}
\end{align}
for all $1 \leq r < s < \infty$.

For many practical applications, 
	the limit cases \(p\in\{1,\infty\}\) play a central role. 
	For \(p = \infty\), the number \(\lVert f\rVert_{\Bs_{\pi}^{p}}\) equals the
	\emph{peak value} of the signal \(f\in\Bs_{\pi}^{p}\), a quantity frequently encountered
	in signal and system theory and communications engineering. As indicated in Section~\ref{sec:MotivDigitTwin}, the peak value 
	is one of the two essential parameters in the study of the \emph{PAPR-problem}.
	In the case of \(p = 1\), the \(\Bs_{\pi}^{p}\)-norm yields a characterization
	of \emph{BIBO-stability} in system and filter theory. In particular, the number
	\(\lVert h\rVert_{\Bs_{\pi}^{1}}\) equals the maximum possible peak value of the output of an LTI-system
	with impulse response \(h\in\Bs_{\pi}^{1}\) that is presented with a normalized input signal.
	For \(h \in \Bs_{\pi}^1\), define the operator \(H : \Bs_{\pi}^{\infty} \rightarrow \Bs_{\pi}^{\infty}\)
	according to 
	\begin{align*}
		(H f)(t)= \int_{-\infty}^{\infty} h(t-\tau)f(\tau)\di{\tau}, \quad t\in\R 
	\end{align*}
	The operator \(H\) characterizes the behavior of an LTI system with impulse response \(h\).
	The BIBO-norm 
	\begin{align*}
		\lVert H \rVert_{\mathrm{BIBO}} := \sup_{\|f\|_{\Bs_{\pi}^{\infty}} = 1} \lVert H f\rVert_{\Bs_{\pi}^{\infty}},
	\end{align*}
	which equals the maximum possible output peak value of \(H\) for normalized input signals, then satisfies 
	\begin{align*}	
		\lVert H \rVert_{\mathrm{BIBO}} = \int_{-\infty}^{\infty} |h(t)| \di{t}.
	\end{align*}
	Furthermore, we have
	\begin{align*}
		(H f)(k) = \sum_{l= -\infty}^{\infty} h(k-l)f(l), \quad k\in\Z.
	\end{align*}
	Since the sequences of sampling values \(\sq{f(k)}_{k\in\Z}\) and \(\sq{h(k)}_{k\in\Z}\) are the standard 
	way of describing \(h\) and \(f\) in digital signal processing, it would be convenient to have an algorithm 
	that computes \(\lVert H \rVert_{\mathrm{BIBO}}\) based on \(\sq{h(k)}_{k\in\Z}\).

A fundamental result in the theory of bandlimited signals is the \emph{Plancherel--P\'olya
theorem} \cite[\emph{Theorem}~3, p.~152]{levin96_book}, which relates the elements of \(\Bs_\pi^p\), \(1 < p < \infty\)
to the elements of \(\ell^p\) by means of an interpolation series. This interpolation series is based
on the \(\sinc\)-function,
\begin{align*}
	\sinc(z) := \begin{cases}	\frac{\sin(\pi z)}{\pi z}	&\text{if}~ z\neq 0,\\
								1							&\text{if}~ z = 0,,
				\end{cases},	
				\quad z\in\C.
\end{align*} 
\begin{theoremnonumber}[Plancherel--P\'olya]
	Let \(1 < p < \infty\). For all sequences $\sq{x(k)}_{k \in \Z} \in \ell^p$, $1 < p < \infty$, there exists
	a unique signal \(f\in\Bs_\pi^p\), such that
	\begin{align*}
		\lim_{N \to\infty} \int_{\R}\left| f(t) - \sum_{k=-N}^{N} x(k) \sinc(t-k) \right|^{p}\di{t}
		= 0
	\end{align*}
	is satisfied. In particular, \(f\) is the unique solution to the interpolation problem
	\(f(k) = x(k)\), \(k\in\Z\). Conversely, for all signals \(f\in\Bs_\pi^p\), $1 < p < \infty$,
	the sequence \(\sq{f(k)}_{k\in\Z}\) belongs to \(\ell^p\) and there exist constants $C_L = C_L(p)>0$ and
	$C_R = C_R(p)>0$, independent of \(f\), such that
	 \begin{equation*}
    C_L \sum_{k=-\infty}^{\infty} \lvert f(k) \rvert^p 
    \leq
    \lVert f \rVert_{\Bs_\pi^p}^p
    \leq
    C_R  \sum_{k=-\infty}^{\infty} \lvert f(k) \rvert^p 
  \end{equation*}
  holds true.
\end{theoremnonumber}

For \(1< p < \infty\), the Plancherel--P\'olya theorem yields a convenient relation between 
the spaces \(\Bs_{\pi}^{p}\) and \(\ell^p\). 
According to the Plancherel--P\'olya theorem, the integers form a set of uniqueness
for signals $f \in \Bs_{\pi}^{p}$, $1< p< \infty$, i.e., \(f\) is uniquely determined by 
the sequence \(\sq{f(k)}_{k\in\Z} \in \ell^p\). This holds true for
\(p\in \{1, \infty\} \) as well. In particular, we have
\begin{align}
	f \equiv 0	\quad\Leftrightarrow\quad \sq{f(k)}_{k\in\Z} \equiv 0	
\end{align}
for all signals $f \in \Bs_{\pi}^{p}$, \(p\in\{1,\infty\}\). For \(p = 1\), this can easily be deduced from
the inclusion \(\Bs_{\pi}^{1} \subsetneq \Bs_{\pi}^{s}\) for \(s > 1\). The equivalence then 
follows by application of the Plancherel--P\'olya theorem. However, it should be noted that there exist sequences
\(\sq{c_k}_{k\in\Z} \in \ell^1\) such that \emph{no} function \(f\in\Bs_{\pi}^{1}\) satisfies the interpolation condition \(\sq{c_k}_{k\in\Z} = \sq{f(k)}_{k\in\Z}\). In other words, the inclusion
\begin{align*}
	\big\{\sq{f(k)}_{k\in\Z} : f \in \Bs_{\pi}^1 \big\} \subset \ell^1
\end{align*} 
is proper. For details, we refer to \cite[\emph{Lecture}~21, p.~155-162]{levin96_book} and \cite[\emph{Chapter}~6, p.~48-66]{higgins96_book}.

For convenience, we introduce the \emph{sampling operator} \setlength{\arraycolsep}{2pt}
\begin{equation*}
  \SOp{p} :\Bs_{\pi}^{p} \rightarrow  \ell^p,~f \mapsto \sq{f(k)}_{k \in \Z} 
\end{equation*}
and its inverse \(\SOp{p}^{-1} =: \TOp{p}\), which we refer to as \emph{interpolation operator}. 

For \(f \in \Bs_{\pi}^{p}\), \(1\leq p\leq \infty\), we have 
\(	\lVert \SOp{p}f \rVert_{\ell^p}
	\leq
	\Col{const:sampling-bounded}(p) \lVert f \rVert_{p}
\), 
where \(\Cor{const:sampling-bounded}(p) \) is independent of \(f\). 
Hence \(\SOp{p}\) is a bounded operator on all considered signal spaces. Since \(\SOp{p}\) is injective,
its inverse \(\TOp{p}\) is well-defined on the linear subspace
\begin{align*}
	\img\big(\SOp{p}\big) =: \mathrm{dom}\big(\TOp{p}\big) 
\end{align*}
of \(\Bs_{\pi}^{p}\). However, unlike \(\SOp{p}\), the operator \(\TOp{p}\) is generally unbounded, 
and the subspace \(\mathrm{dom}\big(\TOp{p}\big)\) is generally not closed.
In Sections~\ref{sec:UncompInfty} and~\ref{sec:UncompOne}, the unboundedness of \(\TOp{\infty}\) and \(\TOp{1}\), respectively, will play a crucial role with regards to computability.

\begin{remark}
	For \(\Bs_{\pi}^{\infty}\) and \(\ell^\infty\) in the sense of \emph{our} definition (c.f. Remark~\ref{rem:BInfty}), 
	the mapping \(\SOp{\infty} : \Bs_{\pi}^{\infty} \rightarrow \ell^\infty\) is one-to-one. This is \emph{not} the case for the relation between the (more general) spaces 
	of bounded entire functions of exponential type and bounded complex-valued sequences, i.e., the requirement of asymptotic decay to zero is necessary
	for the sampling operator to be injective.
\end{remark}

\section{Computable Bandlimited Signals} \label{sec:CBsignals}

In order to study the question of computability, we employ the theory of Turing machines. 
A Turing machine is an abstract device that manipulates symbols on an infinite working tape
according to a finite set of rules \cite{turing36,turing37}. The working tape represents the memory of
the machine, while the set of rules correspond to its program.
Although the concept is very simple, a Turing machine is theoretically capable of simulating any 
algorithm that can be implemented on a real world digital hardware. Hence, if a certain algorithmic problem
cannot be solved on a Turing machine, it can definitely not be solved on an actual computer.
Computability is a mature topic in computer science
\cite{weihrauch00_book,pour-el89_book,boolos02_book,avigad14}.
In signal processing, however, this aspect has not received much attention.

A recursive function is a function, mapping natural numbers into natural numbers, that is
built of simple computable functions and recursions. Among others, they were considered in~\cite{Kl36}.
We will not go into details here, for us it is important that the theory of recursive functions is 
equivalent to the theory of Turing machines, i.e., a function mapping natural numbers to natural numbers
is recursive if and only if it can be computed by a Turing machine~\cite{Tu37}.
For further information on recursive functions see for example \cite{soare87_book}.

A set $\Aset \subset \N$ is called recursively enumerable if it is either empty or the
range of a recursive function. 
A set $\Aset \subset \N$ is called recursive if both $\Aset$ and $\N \setminus \Aset$ are recursively
enumerable, which in turn holds true if and only if the \emph{indicator function} 
\begin{align}	\mathds{1}_\Aset :\N \rightarrow \{0,1\},~
				n \mapsto 	\begin{cases}	1	&\text{if}~ n\in\Aset,\\
											0	&\text{otherwise},
							\end{cases}
\end{align}
of \(\Aset\) is a recursive function. Furthermore, if \(\Aset\) is recursively enumerable, there exists a (total) recursive function
\(g_\Aset : \N^2 \rightarrow \{0,1\}\) that satisfies the following for all \(n\in\N\):
\begin{itemize}
	\item There exists a number \(m\in\N\) such that \(g_\Aset(n,m) = 1\) is satisfied if and only if \(n\in\Aset\) holds true.
	\item If \(g_\Aset(n,m) = 1\) holds true for a number \(m\in\N\), then \(g_\Aset(n,k) = 1\) holds true for all \(k\in\N\) that
		satisfy \(k > m\).
\end{itemize}
We call such a function a \emph{runtime function} for \(\Aset\).

Alan Turing introduced the concept of a computable real number in
\cite{turing36,turing37}. Our definition of a computable real number is based on
computable sequences of rational numbers \cite[p.~14]{pour-el89_book}.
\begin{definition}\label{de:computable-sequence-rationals} 
A sequence of rational numbers $\sq{r_n}_{n \in \N}$ is called computable sequence if there
exist (total) recursive functions\linebreak \(g_{\mathrm{si}},g_{\mathrm{nu}},g_{\mathrm{de}} : \N\to\N\) such that 
\begin{equation*}
  r_n
  =
  \frac{(-1)^{g_{\mathrm{si}}(n)} \cdot g_{\mathrm{nu}}(n)}{g_{\mathrm{de}}(n)} 
\end{equation*}
holds true for all \(n\in\N\).
\end{definition}

\begin{definition}\label{de:computable-real}
  A real number $x$ is said to be computable if there exist a computable sequence of
  rational numbers $\sq{r_n}_{n \in \N}$ and a recursive function $\xi \colon \N \to \N$
  such that $\lvert x - r_{n} \rvert \leq 2^{-M}$ holds true for all
  $n,M \in \N$ that satisfy $n \geq \xi(M)$.  
  By $\Rc$, we denote the set of computable real numbers,
  and by $\Cc = \Rc + \iu \Rc$ the set of computable complex numbers.
\end{definition}
The recursive, i.e., computable function $\xi$ allows us to control the
approximation error algorithmically.
This form of convergence, where we have a computable control of the approximation error is
called \emph{effective convergence}. Furthermore, the function \(\xi\) is referred to as
\emph{modulus of convergence} for the sequence $\sq{r_n}_{n \in \N}$.

Note that commonly used numbers like $\e$ and $\pi$ are computable.
A non-computable real number was for example constructed in \cite{specker49}.

A pair \(\mathfrak{x} := (\sq{r_n}_{n\in\N},\xi)\), where \(\sq{r_n}_{n\in\N}\) is a 
computable sequence of rational numbers that, with respect to the recursive modulus of convergence \(\xi\), 
converges effectively towards some real number \(x\in\R\), is referred to as a \emph{standard description}
of the number \(x\). We denote the set of standard descriptions of real numbers by \(\stdR\).
Naturally, the computable real numbers induce an equivalence relation on \(\stdR\), where two standard descriptions
\(\mathfrak{x}\) and \(\mathfrak{x}'\) are equivalent if they are both standard descriptions of the same computable number \(x\in\Rc\).
In this case, we write
\begin{align*}
	[\mathfrak{x}]_\stdR = [\mathfrak{x}']_\stdR \equiv x.
\end{align*} 
Likewise, a pair \(\mathfrak{c} := (\mathfrak{x}_{\Re},\mathfrak{x}_\mathfrak{\Im})\) consisting of two standard descriptions of real numbers
is referred to as a \emph{standard description} of the complex number \(c\), whenever
\(
	\Re(c) \equiv \left[\mathfrak{x}_{\Re}\right]_\stdR
	\) and \(
	\Im(c) \equiv \left[\mathfrak{x}_{\Im}\right]_\stdR
\)
are satisfied. In this case, we write
\begin{align*}
	c \equiv [\mathfrak{c}]_{\stdC}.
\end{align*}
Further, we denote the set of standard descriptions of complex numbers by \(\stdC\).

\begin{definition}
	A complex-valued sequence \(\sq{x(k)}_{k\in\Z}\) is called \emph{elementary computable} 
	if there exists an interval \(\mathcal{I} := \{-L, \ldots, L\}\subset \Z\), \(L\in\N,\) as well as a \((2L + 1)\)-tuple $\sq{c_k}_{k \in \mathcal{I}}$ of computable complex numbers
	such that
	\begin{equation*}
	  x(k)	=	\begin{cases}	c_k		&\text{if}~ k \in \mathcal{I}, \\
								0		&\text{otherwise},
				\end{cases}
	\end{equation*}
	holds true for all \(k\in\N\).
\end{definition}

\begin{definition}
A sequence $\sq{x(k)}_{k \in \Z}$ in $\ell^p$, \(p\in\Rc\), $1 \leq p \leq \infty$, is called	
computable in $\ell^p$ if there exist a computable sequence $\sq{x_n}_{n \in \N}$,
of elementary computable sequences \(x_n\) and a recursive function $\xi' \colon \N \to \N$, such
that  
\(
	\lVert x - x_{n} \rVert_{\ell^p} \leq 2^{-M}
\)
holds true for all\linebreak $n,M \in \N$ that satisfy $n \geq \xi'(M)$. 
\end{definition}

We denote the set of all sequences that are computable in $\ell^{\infty}$ by $\Cczero$.
Similarly, we denote by $\Cell^p$ the set of all sequences that are computable in $\ell^p$.
Furthermore, we denote by \(\mathscr{X}^p\) the set of all pairs \(\mathfrak{X} = (\sq{x_n}_{n \in \N}, \xi')\)
such that \(\sq{x_n}_{n \in \N}\) is a computable sequence of elementary computable sequences in \(\ell^p\) that
converges effectively towards some \(x\in\ell^p\) with respect to the \(\ell^p\)-norm and the recursive modulus of convergence 
\(\xi'\). In this case, we write 
\(
	\left[\mathfrak{X}\right]_{\mathscr{X}}^p \equiv x.
\)

Next, we define computable bandlimited signals using the same definition as in
\cite{boche20c,boche20d,boche20e,boche20b} that is based on the finite Shannon sampling
series as an interpolation function.
\begin{definition}\label{de:elementary-computable}
A complex-valued function $f$ on \(\C\) is called \emph{elementary computable} 
if there exists an interval \(\mathcal{I} := \{-L, \ldots, L\}\subset \Z\), \(L\in\N,\) as well as a \((2L + 1)\)-tuple $\sq{c_k}_{k\in\mathcal{I}}$ of computable complex numbers
such that
\begin{equation*}
  f(z)
  =
  \sum_{k\in\mathcal{I}} c_k \sinc(z-k)
\end{equation*}
holds true for all \(z\in\C\).
\end{definition}

The building blocks of an elementary computable function are $\sinc$ functions. 
Hence, elementary computable functions are exactly those functions that can be represented
by a finite Shannon sampling series with computable coefficients $\sq{c_k}_{k\in\mathcal{I}}$.
Note that every elementary computable function $f$ is a finite sum of computable
continuous functions and hence a computable continuous function.
As a consequence, for every $t \in \Rc$ the number $f(t)$ is computable.
Further, the sum of finitely many elementary computable functions is elementary
computable, as well as the product of an elementary computable function with a computable
number.

\begin{definition}\label{def:ComputableSignal}
A signal $f \in \Bs_{\pi}^p$, \(p\in\Rc\), $1 \leq p \leq \infty$, is called computable in
$\Bs_{\pi}^p$ if there exists a computable sequence $\sq{f_n}_{n \in \N}$ of elementary computable functions
\(f_n\) and a recursive function $\xi \colon \N \to \N$, such that 
\(
	\lVert f- f_{n} \rVert_{\Bs_\pi^p} \leq 2^{-M}
\)
holds true for all\linebreak $n,M \in \N$ that satisfy $n \geq \xi(M)$. 
\end{definition}

We denote the set of all signals that are computable in $\Bs_{\pi}^{\infty}$ by $\CBs_{\pi}^{\infty}$.
Similarly, we denote by $\CBs_{\pi}^p$ the set of all signals that are computable in $\CBs_{\pi}^p$.
Furthermore, we denote by \(\mathscr{F}^p\) the set of all pairs \(\mathfrak{F} = (\sq{f_n}_{n \in \N}, \xi)\)
such that \(\sq{f_n}_{n \in \N}\) is a computable sequence of elementary computable signals in \(\Bs_{\pi}^p\) that 
converges effectively towards some \(f\in\Bs_{\pi}^p\) with respect to the \(\Bs_{\pi}^{p}\)-norm
and the recursive modulus of convergence \(\xi\). In this case, we write 
\(
	\left[\mathfrak{F}\right]_{\mathscr{F}}^p \equiv f.
\)

According to Definition \ref{def:ComputableSignal}, we can approximate any signal $f \in \CBs_{\pi}^p$,
\(p\in\Rc\), \(1 \leq p \leq \infty\), by an elementary computable signal, where we have an
``effective'', i.e. 
computable control of the approximation error.
For every prescribed approximation error $1/2^M$, $M \in \N$, we can compute
an index $M \in \N$ such that the approximation error
$\lVert f - f_{n} \rVert_{\Bs_{\pi}^p}$ is less than or equal to $1/2^M$ for all
$n \geq \xi(M)$.

\begin{remark}\label{rem:CompNorm}
	As indicated before, the spaces \(\Bs_{\pi}^{p}\) and \(\ell^p\) are Banach spaces for \(1 \leq p \leq \infty\) when
	equipped with the corresponding \(p\)-norm. In particular, the \(p\)-norm is a Turing computable function
	on \(\CBs_{\pi}^p\) and \(\Cell^p\). That is, there exist Turing machines
	\begin{align*}
		\TM_{\mathscr{F}}^p : \mathscr{F}^p \rightarrow \stdR \quad \text{and} \quad
		\TM_{\mathscr{X}}^p : \mathscr{X}^p \rightarrow \stdR
	\end{align*}
	such that
	\begin{align*}
		\big[\TM_{\mathscr{F}}^p(\mathfrak{F}) \big]_{\stdR} \equiv \big\lVert [\mathfrak{F}]_{\mathscr{F}}^p \big\rVert_{\Bs_{\pi}^{p}} \quad \text{and} \quad
		\big[\TM_{\mathscr{X}}^p(\mathfrak{X}) \big]_{\stdR} \equiv \big\lVert [\mathfrak{X}]_{\mathscr{X}}^p \big\rVert_{\ell^{p}}
	\end{align*}
	holds true for all \(\mathfrak{F}\in\mathscr{F}^p\) and all \(\mathfrak{X}\in\mathscr{X}^p\). For details, we refer to~\cite{BoMo21}.
\end{remark}

\section{Conversion of Computable Continuous-Time and Discrete-Time Signals}\label{sec:Conv}

For $f \in \CBs_{\pi}^p$, \(p\in\Rc\), $1\leq p \leq \infty$, let $\sq{f_n}_{n \in \N}$ be a computable sequence of elementary computable functions in \(\CBs_{\pi}^p\)
that converges effectively towards \(f\), with respect to some recursive modulus of convergence \(\xi : \N \rightarrow\N\). 
Analogously, for \(x\in\Cell^p\), \(p\in\Rc\), \(1\leq p \leq \infty\), let \(\sq{x_n}_{n\in\N}\) be a computable sequence of elementary computable sequences that converges effectively towards \(x\), 
with respect to some recursive modulus of convergence \(\xi' : \N \rightarrow\N\).
Observe the following:
\begin{itemize}
	\item The pair \((\sq{f_n}_{n\in\N},\xi) =: \mathfrak{F}\) is a machine-readable description of the signal  \(f\in\CBs_{\pi}^p\) in the language \(\mathscr{F}^p\).
	\item If \(\SOp{p}f = x\) is satisfied, the sequence \(x\in\Cell^p\) determines the signal \(f\in\CBs_{\pi}^p\) uniquely. Hence, in this case, the pair 
		\((\sq{x_n}_{n\in\N},\xi') =: \mathfrak{X}\) is a machine-readable description of the signal \(f\in\CBs_{\pi}^p\) in the language \(\mathscr{X}^p\).
\end{itemize}

From a practical point of view, \(\mathfrak{F}\) can be understood as a computable
\emph{continuous-time} description of \(f\). Furthermore, assuming \(\SOp{p}f = x\) is satisfied, \(\mathfrak{X}\) 
constitutes a computable \emph{discrete-time} description of \(f\). 
Analytically, both descriptions are equivalent, in the sense that they uniquely characterize the same signal. 
However, the analytic equivalence does \emph{not} a priori
imply the computability of the sampling operator \(\SOp{p}\) and its inverse \(\TOp{p}\) in the following sense:
\begin{itemize}
	\item For \(p\in\Rc\), $1\leq p \leq \infty$, we call \(\SOp{p} : \Bs_{\pi}^p \rightarrow \ell^p\) computable
		if there exists a Turing machine \(\TMSOp{p} : \mathscr{F}^p\rightarrow\mathscr{X}^p\) such that
		\begin{align*}
			\SOp{p} \big[ \mathfrak{F} \big]_{\mathscr{F}}^p \equiv \big[ \TMSOp{p} (\mathfrak{F}) \big]_{\mathscr{X}}^p
		\end{align*}
		is satisfied for all \(\mathfrak{F} \in \mathscr{F}^p\).
		That is, \(\TMSOp{p}\) returns a computable discrete-time description \(\mathfrak{X}\)
		of \(f\in\CBs_{\pi}^p\) whenever it is presented with a computable continuous-time description \(\mathfrak{F}\) of \(f\) as input.
	\item For \(p\in\Rc\), $1\leq p \leq \infty$, we call \(\TOp{p} : \img\big(\SOp{p}\big) \rightarrow \Bs_{\pi}^p\) computable
		if there exists a Turing machine \(\TMTOp{p} : \mathscr{X}^p\rightarrow\mathscr{F}^p\) such that
		\begin{align*}
			\TOp{p} \big[ \mathfrak{X} \big]_{\mathscr{X}}^p \equiv \big[ \TMTOp{p} (\mathfrak{X}) \big]_{\mathscr{F}}^p
		\end{align*}
		is holds true for all \(\mathfrak{X} \in \mathscr{X}^p\) that satisfy \(\SOp{p}f \equiv \big[ \mathfrak{X} \big]_{\mathscr{X}}^p\)
		for some \(f\in \CBs_{\pi}^p\).
		That is, \(\TMTOp{p}\), returns a computable continuous-time description \(\mathfrak{F}\)
		of \(f\in\CBs_{\pi}^p\) whenever it is presented with a computable discrete-time description \(\mathfrak{X}\) of \(f\) as input.
\end{itemize}

Let \(f\) be a signal in \(\CBs_{\pi}^p\), \(p\in\Rc\), \(1 \leq p \leq \infty\),
and assume that \(\mathfrak{F} = (\sq{f_n}_{n\in\N},\xi)\) is a computable continuous-time description thereof. 
By definition, \(\sq{f_n}_{n\in\N}\) is characterized by a computable sequence \(\sq{\sq{c_{n,k}}_{k \in \mathcal{I}(n)}}_{n\in\N}\) of tuples of 
computable complex numbers, where \(\mathcal{I}(n) = \{-L(n), \ldots, L(n)\}\) is a computable interval
of natural numbers, that satisfy
\begin{equation*}
  f_n(z)
  =
  \sum_{k \in \mathcal{I}(n)} c_{n,k} \sinc(z-k)
\end{equation*}
for all \(n\in\N, z\in\C\). Next, define \(x := \SOp{p}f\) as well as \(\sq{x_n}_{n\in\N} := \sq{\SOp{p}f_n}_{n\in\N}\). Then, for all \(n\in\N\), we have 
\begin{align*}
	x_n(k) = 	\begin{cases}	c_{n,k}	&\text{if}~ k\in\mathcal{I}(n),\\
													0				&\text{otherwise},
						\end{cases}
\end{align*}
for all \(k\in\Z\). Hence, \(\sq{x_n(k)}_{k\in\Z}\) is an elementary computable sequence for all \(n\in\N\). 
With \(\Col{const:log-sampling-bounded} := \lceil \log_2 \Cor{const:sampling-bounded} \rceil\), we have 
\begin{align*}
  2^{\Cor{const:log-sampling-bounded}-M} 
		&\geq	\Cor{const:sampling-bounded} \lVert f-f_n \rVert_{\Bs_\pi^p} \\
		&\geq	\lVert \SOp{p}(f - f_n) \rVert_{\ell^p} \\
		&= \lVert (\SOp{p}f) - (\SOp{p}f_n) \rVert_{\ell^p} \\
		&= \lVert x - x_n \rVert_{\ell^p}
\end{align*}
for all \(n,M \in \N\) that satisfy \(n \geq \xi(M)\). 
Define the function \(\xi' :~ \N \rightarrow \N,~ M \mapsto \xi(\Cor{const:log-sampling-bounded} + M)\)
and observe that \(\xi'\) is recursive. We have
\begin{align*}
	\lVert x - x_n \rVert_{\ell^p} \geq 2^{-M}
\end{align*}
for all \(n,M\in\N\) that satisfy \(n\geq \xi'(M)\). Hence, the sequence \(\sq{x_n}_{n\in\N}\) converges effectively
towards \(x\in\ell^p\), with respect to the recursive modulus of convergence \(\xi'\). In other words, \(\mathfrak{X} = (\sq{x_n}_{n\in\N},\xi')\)
is a computable discrete-time description of \(f\).

The previous paragraph explicitly characterizes an algorithm that, given a computable continuous-time description of a signal \(f\in\CBs_{\pi}^p\)
as input, returns a computable discrete-time description of the same signal. Hence, the sampling operator \(\SOp{p}\) is computable for all \(p\in[1,\infty]\cap\Rc\).

In digital signal processing, a signal \(f\in\CBs_{\pi}^{p}\), \(p\in\Rc\), \(1\leq p\leq \infty\) is usually characterized by a discrete-time description 
\(\mathfrak{X}\). 
In contrast, there are many reasons why a continuous-time description \(\mathfrak{F}\) can be beneficial. We will discuss some of them in the following.
\begin{itemize}
	\item\textit{time concentration}. For a signal \(f\in\CBs_{\pi}^{p}\), \(p\in\Rc\), \(1\leq p \leq \infty\) and a number \(L\in\Rc\), the mapping
		\begin{align}	\label{eq:timeconcentration}
			L \mapsto 	\begin{cases}	\int_{-L}^{L} |f(t)|^p \di{t}	&\text{if}~1\leq p < \infty,\\
										\max_{-L\leq t\leq L} |f(t)|		&\text{if}~p=\infty,
						\end{cases}
		\end{align}
		serves as an indicator of the size of that ``portion'' of \(f\) which is located in the interval \([-L,L]\), and is
		referred to as \emph{time concentration} of \(f\). Given a continuous-time description \(\mathfrak{F}\), 
		\eqref{eq:timeconcentration} can be evaluated algorithmically. For details, we again refer to~\cite{BoMo21}. 
	\item\textit{\(p\)-norm}. As indicated in Remark~\ref{rem:CompNorm}, the \(p\)-Norm of \(f\) can directly be calculated 
		from any continuous-time description \(\mathfrak{F}\) of \(f\).
	\item\textit{time-derivative}. Among other things, the time-derivative \(\dot{f}\) of a signal \(f\in\Bs_{\pi}^{p},~ 1\leq p\leq \infty\), 
		is essential for estimating its dynamics by means of the mean value theorem. We have
		\begin{align*}
			\big\lVert \dot{f}\big\rVert_{\Bs_{\pi}^{p}} \leq \pi \big\lVert  f \big\rVert_{\Bs_{\pi}^{p}}.
		\end{align*}
		Hence, if \((\sq{f_n}_{n\in\N},\xi)\) is a computable continuous-time description of \(f\), we further have
		\begin{align*}
			\big\lVert \dot{f} - \dot{f_n} \big\rVert_{\Bs_{\pi}^{p}} \leq \pi \big\lVert  f -f_n \big\rVert_{\Bs_{\pi}^{p}}
			\leq \frac{\pi}{2^{-M}} < \frac{1}{2^{2-M}}.
		\end{align*}	
		Consequently, from any computable continuous-time description of \(f\), we can directly determine a 
		computable continuous-time description of the time-derivative \(\dot{f}\).
\end{itemize}
Accordingly, the question arises whether the interpolation operator is computable. The remainder of the article will address this problem.

\section{Uncomputability of the Interpolation Operator for \MakeLowercase{\(p = \infty\)}}\label{sec:UncompInfty}
In this section, we consider the limit case of \(p = \infty\). We will prove that the 
associated interpolation operator \(\TOp{\infty} : \img\big(\SOp{\infty}\big) \rightarrow \Bs_{\pi}^\infty\) is 
not computable (in the sense of Section \ref{sec:Conv}):

\begin{theorem}\label{thm:UncompInfty}
  The interpolation operator \(\TOp{\infty} : \img\big(\SOp{\infty}\big) \rightarrow \Bs_{\pi}^\infty\) is uncomputable.
\end{theorem}

In essence, the uncomputability of \(\TOp{\infty}\) is a consequence of its discontinuity.
In the following, we will establish two preliminary lemmas. Afterward, we provide a proof for Theorem~\ref{thm:UncompInfty}.

	\begin{lemma}\label{lem:bo1}
		There exists a computable sequence $\sq{f_n}_{n \in \N}$ of elementary computable functions such that
		\begin{align*}	f_n \left( \sfrac{1}{2} \right) = 1,\quad
						\lVert \SOp{\infty}f_n \rVert_{\ell^{\infty}} < \sfrac{1}{n},\quad\text{and}\quad
						\lVert f_n \rVert_{\Bs_{\pi}^{\infty}} \leq \Col{const:le:bo1}
		\end{align*}
		are satisfied for all \(n\in\N\).
	\end{lemma}
	Observe that Lemma \ref{lem:bo1} already implies the unboundedness (and hence discontinuity) of the interpolation
	operator \(\TOp{\infty}\), since
	\begin{align}	\label{eq:DiscT}
		\lim_{n\to\infty} \frac{\lVert \TOp{\infty}\SOp{\infty}f_n \rVert_{\Bs_{\pi}^{\infty}}}{\lVert \SOp{\infty}f_n \rVert_{\ell^{\infty}}}
			= \lim_{n\to\infty} \frac{\lVert f_n \rVert_{\Bs_{\pi}^{\infty}}}{\lVert \SOp{\infty}f_n \rVert_{\ell^{\infty}}}
			> \lim_{n\to\infty} \frac{1}{\sfrac{1}{n}} = \infty
	\end{align}
	holds true. The discontinuity of \(\TOp{\infty}\) is one of the core ingredients in proving its uncomputability.
	
	\begin{proof}[Proof of Lemma \ref{lem:bo1}]
		To begin with, consider the function \(g : \R\times\N \rightarrow \R\) defined according to
		\begin{align*}
		  g(t,N) :	&= \frac{1}{C(N)} \sum_{k=1}^{N} (-1)^k \sinc(t-k), \\
		  C(N) :	&= -\frac{1}{\pi} \sum_{k=1}^{N} \frac{1}{k-\frac{1}{2}}.  
		\end{align*}
		and set \(f_n(t) := g\big(t, 2^{8n}\big)\) for all \(n\in\N,t\in\R\). Then, $\sq{f_n}_{n \in \N}$ is a 
		computable sequence of elementary computable functions, and hence a computable sequence of 
		functions in $\CBs_{\pi}^{\infty}$.
		
		In the following, we want to prove that $\sq{f_n}_{n \in \N}$ satisfies the properties
		required by the lemma. First, observe that 
		\begin{align}
			\sinc\left(\frac{1}{2} - k\right) = \frac{(-1)^k}{\pi(\frac{1}{2}-k)}
		\end{align}
		holds true for all \(k\in\N\). Consequently, for all \(n\in\N\), we have
	\begin{align*}	f_n(\sfrac{1}{2}) 
					&=	\left(-\frac{1}{\pi} \sum_{k=1}^{2^{8n}} \frac{1}{k-\frac{1}{2}}\right)^{-1}
						\sum_{k=1}^{2^{8n}} (-1)^k \frac{(-1)^k}{\pi(\frac{1}{2}-k)} \\
					&=	\left( \sum_{k=1}^{2^{8n}} \frac{1}{\frac{1}{2}-k}\right)^{-1}
						\sum_{k=1}^{2^{8n}} \frac{1}{(\frac{1}{2}-k)} = 1.
	\end{align*}
	Next, observe that for all \(N\in\N\), the inequality
	\begin{align*}
	  \lvert C(N) \rvert
	  &=
	  \frac{1}{\pi} \sum_{k=1}^{N} \frac{1}{k-\frac{1}{2}}
	  >
		\frac{1}{\pi} \sum_{k=1}^{N} \int_{k}^{k+1} \frac{1}{\tau-\frac{1}{2}} \di{\tau} \\
	  &=
		\frac{1}{\pi} \int_{1}^{N+1} \frac{1}{\tau-\frac{1}{2}} \di{\tau} \\
	  &=
		\frac{1}{\pi} \left(\ln\left(N+\frac{1}{2}\right) - \ln\left(\frac{1}{2}\right)\right) \\
	  &>
		\frac{\log_2(N)}{4} 
	\end{align*}
	is satisfied. Furthermore, we have
	\begin{align*}
	  \lVert \SOp{\infty}f_n \rVert_{\ell^{\infty}}
	  &=
	   \sup_{m\in\N} \left| \frac{1}{ C(2^{8n}) } \sum_{k=1}^{2^{8n}} (-1)^k \sinc(m-k)\right| \\
	  &\leq \frac{1}{ |C(2^{8n})| }  \sup_{m\in\N}  \sum_{k=1}^{2^{8n}} \left|(-1)^k \sinc(m-k)\right| \\
	  &= \frac{1}{ |C(2^{8n})| }  \sup_{m\in\N}  \sum_{k=1}^{2^{8n}} \mathds{1}_{m}(k) \\
	  &= \frac{1}{ |C(2^{8n})| } 
	  \end{align*}
	  for all \(n\in\N\), where \(\mathds{1}_{m} : \N \rightarrow \{0,1\}\) is the indicator function of the singleton set \(\{m\}\subset\N\).
	  Consequently,
	  \begin{align*}
	  \lVert \SOp{\infty}f_n \rVert_{\ell^{\infty}}<
	  \frac{4}{\log_2(2^{8n})}
	  =
	  \frac{1}{2 n \log_2(2)}
	  <
	  \frac{1}{n} 
	\end{align*}
	is satisfied for all \(n\in\N\).
	It remains to show that the sequence \(\sq{f_n}_{n\in\N}\) satisfies \(\lVert f_n \rVert_{\Bs_{\pi}^{\infty}} \leq \Cor{const:le:bo1}\)
	for all \(n\in\N\). For $t \in \Z$ and $N \in \N$, $N \geq 1$, we have
	\begin{align*}
	  \lvert g(t,N) \rvert
	  \leq
	  \frac{1}{\lvert C(N) \rvert}
	  <
	  \frac{4}{\log_2(N)} .
	\end{align*}
	  Furthermore, for $t \in \R \setminus \Z$, we have
	  \begin{align*}
		&\left| \sum_{k=1}^{N} (-1)^k \frac{\sin(\pi(t-k))}{\pi(t-k)} \right|
		\leq
		  \sum_{k=1}^{N} \left| \frac{\sin(\pi(t-k))}{\pi(t-k)} \right| \notag \\
		&\quad<
		  2 + \frac{1}{\pi} \sum_{k=1}^{k_1(t)} \frac{1}{t-k} + \frac{1}{\pi}
		  \sum_{k=k_2(t)}^{N} \frac{1}{k-t} \notag \\
		&\quad<
		  2 + \frac{1}{\pi} \sum_{k=1}^{k_1(t)} \frac{1}{k_1(t)+1-k}
		 + \frac{1}{\pi} \sum_{k=k_2(t)}^{N} \frac{1}{k-k_2(t)+1} \notag \\
		&\quad=
		  2 + \frac{1}{\pi} \sum_{k=1}^{k_1(t)} \frac{1}{k} + \frac{1}{\pi}
		  \sum_{k=1}^{N-k_2(t)+1} \frac{1}{k} \notag \\
		&\quad\leq
		  2 + \frac{2}{\pi} \sum_{k=1}^{N} \frac{1}{k} \notag \\
		&\quad\overset{\text{(a)}}{<}
		  2 + \frac{2}{\pi} + \frac{2}{\pi} \log_2(N) ,
	  \end{align*}
	  where $k_1(t)$ is the largest natural number that is smaller than or equal to $N$ and
	  satisfies $k_1(t) + 1 < t$. 
	  If no such number exists, then the sums above involving $k_1(t)$ are the empty sums.
	  Furthermore, $k_2(t)$ is the smallest natural number such that $k_2(t)-1 > t$ holds true.
	  If $k_2(t) > N$ is satisfied, then the sums above involving $k_2(t)$ are the empty sums.
	  Moreover, (a) follows from the inequality
	  \begin{align*}
		\sum_{k=1}^{N} \frac{1}{k}
		<
		1 + \sum_{k=2}^{N} \int_{k-1}^{k} \frac{1}{\tau} \di{\tau}
		=
		  1 + \int_{1}^{N} \frac{1}{\tau} \di{\tau}
		=
		  1 + \ln(N). 
	  \end{align*}
	It follows that there exists a constant $~\Cor{const:le:bo1}$, such that
	\begin{align*}
	  \lvert g(t,N) \rvert
	  &=
	  \frac{1}{\lvert C(N) \rvert} \left| \sum_{k=1}^{N} (-1)^k
		  \frac{\sin(\pi(t-k))}{\pi(t-k)} \right| \\
	  &\leq
		\frac{8 + \frac{8}{\pi} + \frac{8}{\pi} \log_2(N)}{\log_2(N)} \leq
	  \Cor{const:le:bo1}
	\end{align*}
	is satisfied for all $t \in \R$ and all $N \in \N$, $N \geq 1$.
	Hence, the sequence \(\sq{f_n}_{n\in\N}\) satisfies \(\lVert f_n \rVert_{\Bs_{\pi}^{\infty}} \leq \Cor{const:le:bo1}\)
	for all \(n\in\N\), which concludes the proof.
	\end{proof}

\begin{lemma}\label{lem:bo2}
	Let $\Aset \subset \N$ be a recursively enumerable set. There exists a sequence
	$\sq{f_m^*}_{m \in \N}$ of elementary computable functions that satisfies the following:
	\begin{enumerate}
		\item The sequence \(\sq{x_m^*}_{m\in\N} = \sq{\SOp{\infty}f_m^*}_{m \in \N}\) is a computable sequence of sequences in $\Cczero$.
		\item We have \(f_m^*(\sfrac{1}{2}) = \mathds{1}_{\Aset}(m)\) for all \(m\in\N\).
	\end{enumerate}
\end{lemma}\begin{proof}
	Let $\Aset \subset \N$ be a recursively enumerable set with runtime function
	\(g_\Aset : \N^2 \rightarrow \{0,1\}\). Consider the function \(h : \N^2 \rightarrow \N\)
	defined according to 
	\begin{align*}
		h(m,k) := \sum_{l=0}^{2^{k+2}} (1 - g_\Aset(m,l)).
	\end{align*}
	Further, let \(\sq{f_n}_{n\in\N}\)
	be a computable sequence of elementary computable functions as specified by Lemma \ref{lem:bo1}
	and define
	\begin{align*}
			f_{m,k} := f_{h(m,k)}
	\end{align*}
	for all \(m,k\in\N\). Then, \(\sq{f_{m,k}}_{m,k\in\N}\) is a computable double sequence of elementary computable functions.
	Moreover, the sequence \(\sq{x_{m,k}}_{m,k\in\N} := \sq{\SOp{\infty}f_{m,k}}_{m,k\in\N}\) is a computable double sequence of 
	elementary computable sequences.
	
	For $m \in \Aset$, there exists \(k\in\N\) such that for all \(l\in\N\) that satisfy \(l\geq k\), we have
	\(
	f_{m,l} = f_{m,k},
	\)
	i.e., the limit value \(\lim_{l\to\infty} f_{m,l}\) exists and is an elementary computable function.
	Furthermore, there exists an \(n\in\N\) such that \(f_n = \lim_{l\to\infty} f_{m,l}\) is satisfied.
	We define
	\begin{align*}
		f_m^* = 			\begin{cases}		\lim_{l\to\infty} 	f_{m,l} 	&\quad \text{if}~ m \in \Aset, \\
																	0			&\quad \text{otherwise},
							\end{cases}
	\end{align*}
	for all \(m\in\N\). Hence, \(\sq{f_m^*}_{m\in\N}\) is a sequence of elementary computable functions.
	Furthermore, the sequence \(\sq{x_m^*}_{m\in\N} = \sq{\SOp{\infty}f_m^*}_{m\in\N}\) is a sequence of elementary computable sequences.
	
	In the following, we will prove by case distinction that for all \(m\in\N\), the sequence \(\sq{x_{m,k}}_{k\in\N}\) converges effectively towards \(x_m^*\) in \(\ell^\infty\),
	with respect to the recursive modulus of convergence \(\xi' : \N^2 \rightarrow \N,~ (m,K) \mapsto K\). First, assume that \(m\in\Aset\) is satisfied. 
	Then, there exists \(k\in\N\) such that for all 
	\(K\in\N\) that satisfy \(K\geq k\), we have \(x_{m,K} = x_{m,k}\). Consider the smallest such \(k\in\N\) and observe the following:
	\begin{itemize}
		\item Assume that \(K\in\N\) satisfies \(K\geq k\). Then, we have
			\begin{align*} \lVert x_m^* - x_{m,K} \rVert_{\ell^\infty} = \lVert x_{m,k} - x_{m,k} \rVert_{\ell^\infty} = 0 < \frac{1}{2^K}.
			\end{align*}
		\item Assume that \(K\in\N\) satisfies \(K < k\). Observe that by the properties of the runtime function \(g_{\mathcal{A}}\) and the construction of
			\(h\) as above, \(K < k\) implies \(h(m,k) \geq h(m,K) = 2^{K+2} + 1\). Then, we have
			\begin{align*} \lVert x_m^* - x_{m,K} \rVert_{\ell^\infty} 	&= \lVert x_{m,k} - x_{m,K} \rVert_{\ell^\infty} \\
																		&\leq\lVert x_{m,k} \rVert_{\ell^\infty} + \lVert x_{m,K} \rVert_{\ell^\infty} \\
																		&=\lVert \SOp{\infty} f_{h(m,k)} \rVert_{\ell^\infty} + \lVert \SOp{\infty} f_{h(m,K)} \rVert_{\ell^\infty}\\
																		&\leq \frac{1}{h(m,k)} + \frac{1}{2^{K+2} + 1} \\
																		&\leq \frac{1}{2^{K+2}} + \frac{1}{2^{K+2}} \\
																		&\leq \frac{2}{2^{K+2}} \\ 
																		&= \frac{1}{2^{K+1}}.
			\end{align*}
	\end{itemize}
	Hence, for \(m\in\Aset\), we have \(\lVert x_m^* - x_{m,K} \rVert_{\ell^\infty} < 2^{-K}\) for all \(K\in\N\). Assume now that \(m\in\Aset^{\complement}\) is satisfied
	and observe the following:
	\begin{itemize}
		\item If \(m\in\Aset^{\complement}\) holds true, then \(x_m^* = 0\) is satisfied. Thus, for all \(K\in\N\), we have
			\begin{align*} \lVert x_m^* - x_{m,K} \rVert_{\ell^\infty} = \lVert x_{m,K} \rVert_{\ell^\infty} \leq \frac{1}{2^{K+2}}.
			\end{align*}
	\end{itemize}
	We conclude that \(\lVert x_m^* - x_{m,K} \rVert_{\ell^\infty} < 2^{-K}\) is satisfied for all \(m,K\in\N\). In other words, \(\sq{x_{m,k}}_{k\in\N}\) converges effectively towards 
	\(x_m^*\) in \(\ell^\infty\), with respect to the recursive modulus of convergence \(\xi' : \N^2 \rightarrow \N,~ (m,K) \mapsto K\). 
	Consequently, \(\sq{x_m^*}_{m\in\N}\) is a computable sequence of functions in \(\Cell^\infty\).
	
	It remains to show that the sequence \(\sq{f_m^*}_{m\in\N}\) satisfies \(f_m^*(\sfrac{1}{2}) = \mathds{1}_{A}(m)\) for all \(m\in\N\), which we again prove by case distinction.
	Recall that if \(m\in\Aset\) is satisfied, then \(f_m^*\) is an elementary computable function such that \(f_n = f_m^*\) holds true for some \(n\in\N\). By assumption,
	we have \(f_n(\sfrac{1}{2}) = 1\) for all \(n\in\N\). Hence, if \(m\in\Aset\) holds true, we have  
	\begin{align*}
			f_m^*(\sfrac{1}{2}) = f_n(\sfrac{1}{2}) = 1.
	\end{align*}
	On the other hand, if \(m\in\Aset^{\complement}\) is satisfied, then \(f_m^*\) is the trivial elementary computable function, i.e., we have \(f_m^* = 0\).
	Thus, in this case, 
	\begin{align*}	f_m^*(\sfrac{1}{2}) = 0
	\end{align*}
	holds true. We conclude that \(\sq{f_m^*}_{m\in\N}\) satisfies \(f_m^*(\sfrac{1}{2}) = \mathds{1}_{A}(m)\) for all \(m\in\N\),
	which completes the proof.  	
\end{proof}

\begin{proof}[Proof of Theorem \ref{thm:UncompInfty}]
	We prove the theorem by contradiction. Let \(\Aset \subset \N\) be a recursively enumerable, nonrecursive subset of the natural numbers.
	Further, let \(\sq{f_m^*}_{m\in\N}\) be a sequence of elementary computable functions as specified by Lemma \ref{lem:bo2}. Define
	\(\sq{x_m^*}_{m\in\N} := \sq{\SOp{\infty}f_m^*}_{m\in\N}\) and consider a computable double sequence \(\sq{x_{m,k}}_{m,k\in\N}\) of elementary computable sequences
	such that for all \(m\in\N\), the sequence \(\sq{x_{m,k}}_{k\in\N}\) converges effectively towards \(x_m^*\) in \(\ell^\infty\),
	with respect to the recursive modulus of convergence \(\xi' : \N^2 \rightarrow \N,~ (m,K) \mapsto \xi'(m,K)\). 
	
	Since \(\sq{f_m^*}_{m\in\N}\) is a sequence of elementary computable functions, it is also a sequence of functions in \(\CBs_{\pi}^{\infty}\). Moreover,
	the pair 
	\begin{align*}
		\mathfrak{X}^* := \big(\sq{x_{m,k}}_{m,k\in\N}, \xi'\big)
	\end{align*}
	is a computable discrete-time description of the sequence of functions \(\sq{f_m^*}_{m\in\N}\) in the language \(\mathscr{X}^\infty\). 
	
	Assume now that 
	\(\TOp{\infty}\) is computable in the sense of Section \ref{sec:Conv}. Then,
	there exists a Turing machine \(\TMTOp{\infty}\) such that the pair
	\begin{align*}
		\mathfrak{F}^* = \big(\sq{f_{m,k}}_{m,k\in\N},\xi\big) := \TMTOp{\infty}\big(\sq{x_{m,k}}_{m,k\in\N}, \xi'\big)
	\end{align*}
	is a computable continuous-time description of the sequence \(\sq{f_m^*}_{m\in\N}\) in \(\CBs_{\pi}^{\infty}\). In other words,
	\(\sq{f_{m,k}}_{m,k\in\N}\) is a computable double sequence of elementary computable functions and \(\xi: \N^2 \rightarrow \N, (m,K) \mapsto \xi(m,K)\)
	is a (total) recursive function such that
	\begin{align*}
			\lVert f_m^* - f_{m,k} \rVert_{\Bs_{\pi}^{\infty}} < \frac{1}{2^K}
	\end{align*}
	holds true for all \(m,k,K\in\N\) that satisfy \(k \geq \xi(m,K)\). Further, we have 
	\begin{align*}
		| f_m^*(\sfrac{1}{2}) - f_{m,k}(\sfrac{1}{2}) | < \frac{1}{2^K}
	\end{align*}
	for all \(m,k,K\in\N\) that satisfy \(k \geq \xi(m,K)\), as follows from the properties of the \(\CBs_{\pi}^{\infty}\)-norm.
	
	Observe that for all \(m,k\in\N\), the signal \(f_{m,k}\) is a finite linear combination of computable continuous functions and thus,
	as indicated in Section \ref{sec:CBsignals}, a computable continuous function itself. Consequently, the sequence \(\sq{f_{m,k}(\sfrac{1}{2})}_{m,k}\)
	is a computable double sequence of computable complex numbers. Moreover, there exists computable double sequences \(\sq{r_{m,k}}_{m,k\in\N}\),
	\(\sq{s_{m,k}}_{m,k\in\N}\), of rational numbers, such that
	\begin{align*}
		\big| f_m^*(\sfrac{1}{2}) - (r_{m,k} +\iu s_{m,k}) \big| < \frac{1}{2^K}
	\end{align*}
	holds true for all \(m,k,K\in\N\) that satisfy \(k \geq \xi(m,K)\). 
	
	Recall that by construction, we have \(f_m^*(\sfrac{1}{2}) = \mathds{1}_{\Aset}(m)\) for all \(m\in\N\). It follows that
	\begin{align*}
		m \in \Aset 	\quad \Leftrightarrow \quad 	\big(1 - r_{m,\xi(m,1)}\big)^2 + \big(s_{m,\xi(m,1)}\big)^2 < \frac{1}{4} 
	\end{align*}
	holds true for all \(m\in\N\). In particular, we have
	\begin{align*}
		\Aset = \Bigg\{m \in \N : \underbrace{\big(1 - r_{m,\xi(m,1)}\big)^2 + \big(s_{m,\xi(m,1)}\big)^2 < \frac{1}{4}}_{\equiv: P(m)} \Bigg\}.
	\end{align*}
	Since both \(\sq{r_{m,k}}_{m,k\in\N}\) and \(\sq{s_{m,k}}_{m,k\in\N}\) are computable sequences of rational numbers, the predicate \(P\)
	can be evaluated algorithmically, making \(\Aset\) a recursive set. The latter is a direct contradiction to the assumption, which concludes the proof.
\end{proof}

Furthermore, we obtain the following from the proof of Theorem~\ref{thm:UncompInfty}:

\begin{corollary}
	There does \emph{not} exist a Turing machine\linebreak \(\TM : \mathscr{X}^\infty \times \stdR \rightarrow \stdR\) 
	such that 
	\begin{align}
		\big[\TM(\mathfrak{X},\mathfrak{t})\big]_{\stdR} \equiv f\big(t\big)
	\end{align}
	holds true for all \((f,t) \in \CBs_{\pi}^{\infty}\times \Rc\), \((\mathfrak{X},\mathfrak{t}) \in \mathscr{X}^\infty \times \stdR\), that satisfy
	\((\SOp{\infty}f,t)\equiv ([\mathfrak{X}]_{\mathscr{X}}^\infty,[\mathfrak{t}]_{\stdR})\).
\end{corollary}

Finally, observe that for a recursively enumerable set \(\Aset\) and a sequence 
$\sq{f_m^*}_{m \in \N}$ as specified by Lemma \ref{lem:bo2},
we have 
\begin{align}
		\lVert f_m^* \rVert_{\Bs_{\pi}^{\infty}} 
		\begin{cases}	\geq 	1 	&\text{if}~m\in\Aset,\\
									= 		0		&\text{otherwise}.
		\end{cases}
\end{align}
Hence, we also obtain the following from the proof of Theorem~\ref{thm:UncompInfty}:

\begin{corollary}
	There does \emph{not} exist a Turing machine\linebreak \(\TM : \mathscr{X}^\infty \rightarrow \stdR\) 
	such that 
	\begin{align}
		\big[\TM(\mathfrak{X})\big]_{\stdR} \equiv \lVert f \rVert_{\Bs_{\pi}^{\infty}} 
	\end{align}
	holds true for all \(f \in \CBs_{\pi}^{\infty}\), \(\mathfrak{X} \in \mathscr{X}^\infty\), that satisfy
	\(\SOp{\infty}f\equiv [\mathfrak{X}]_{\mathscr{X}}^\infty\).
\end{corollary}

\section{Uncomputability of the Interpolation Operator for \MakeLowercase{\(p = 1\)}} \label{sec:UncompOne}
In section \ref{sec:UncompInfty}, we have shown the uncomputability of \(\TOp{\infty}\). 
In this section, we investigate the limit case of \(p = 1\). We will prove that the 
associated interpolation operator \(\TOp{1} : \img\big(\SOp{1}\big) \rightarrow \Bs_{\pi}^1\) is 
not computable (in the sense of Section \ref{sec:Conv}) either:

\begin{theorem}\label{thm:UncompOne}
  The interpolation operator \(\TOp{1} : \img\big(\SOp{1}\big) \rightarrow \Bs_{\pi}^1\) is uncomputable.
\end{theorem}

Some of the ideas and proof techniques we have already applied in Section~\ref{sec:UncompInfty} can be extended to the case of \(p=1\),
as we will see in the following. In particular, this concerns the discontinuity of the interpolation operator \(\TOp{1}\).
Again, we start the analysis by establishing several preliminary lemmas. Subsequently, we give a proof of Theorem~\ref{thm:UncompOne}.

\begin{lemma}[\hspace{1sp}\cite{boche20g}]\label{lem:bo3}
	For all \(N\in\N\), the function \(q_N\) defined according to 
	\begin{align*}	
		q_N : \R \rightarrow \R, ~t \mapsto \sinc(t) - \frac{1}{N}\sum_{k=1}^{N} \sinc(t+2k)
	\end{align*}
	satisfies
	\begin{align}	\label{eq:fooI}
					\int_{-\infty}^{\infty} \lvert q_N(t) \rvert \di{t}
					<
					4 + \frac{5}{\pi} \ln(2N +1) 
	\end{align}
	as well as
	\begin{align}	\label{eq:fooII}
					\int_{-\infty}^{\infty} \lvert q_N(t) \rvert \di{t}
					>
					\frac{1}{6 \pi} \ln \left( \frac{N}{2} \right) - \frac{1}{\pi} .
	\end{align}
\end{lemma}\begin{proof}
	For the proof of \eqref{eq:fooI}, we refer directly to \cite[Eq.~18]{boche20g}, where
	the inequality is explicitly derived. A proof of \eqref{eq:fooII} can be deduced from \cite[Eq.~11]{boche20g}
	with a few additional steps. According to \cite{boche20g}, the function \(q_N\)
	satisfies
	\begin{align*}
		\int_{0}^{\infty} q_N(t)a(t)\di{t} > \frac{1}{6\pi}\ln\left(\frac{N}{2}\right) -\frac{1}{\pi}
	\end{align*}
	for all \(N\in\N\), where \(a : \R \rightarrow \R\) is the function defined by
	\begin{align}
		t \mapsto 	\begin{cases}	0				&\text{if}~ t < 0,\\
									t\sin(\pi t)	&\text{if}~ 0 \leq t < 1, \\
									\sin(\pi t)		&\text{otherwise}.
					\end{cases}
	\end{align}
	Consequently, for all \(N\in\N\), we have
	\begin{align*}	\frac{1}{6\pi}\ln\left(\frac{N}{2}\right) -\frac{1}{\pi} 
						&\overset{\text{(a)}}{<} \int_{-\infty}^{\infty} q_N(t)a(t)\di{t} \\
						&\leq \left|\int_{-\infty}^{\infty} q_N(t)a(t)\di{t}\right| \\
						&\leq \int_{-\infty}^{\infty} |q_N(t)||a(t)|\di{t} \\
						&\overset{\text{(b)}}{\leq} \int_{-\infty}^{\infty} |q_N(t)|\di{t},
	\end{align*}
	where (a) follows from the fact that \(a(t) = 0\) for all \(t < 0\) and 
	(b) follows from the fact that \(|a(t)| \leq 1\) for all \(t \in \R\). 
\end{proof}

\begin{lemma}\label{lem:bo4}
	There exists a computable sequence $\sq{f_n}_{n \in \N}$ of elementary computable functions such that
	\begin{align*}	\lVert f_n \rVert_{\Bs_{\pi}^{1}} = 1 \quad\text{and}\quad \lVert \SOp{1}f_n \rVert_{\ell^{1}} < \sfrac{1}{n} 
	\end{align*}
	are satisfied for all \(n\in\N\).
\end{lemma}

Recall that Lemma \ref{lem:bo1} implies the discontinuity of \(\TOp{\infty}\), as follows from \eqref{eq:DiscT}. 
In the same manner, Lemma \ref{lem:bo4} implies the discontinuity of \(\TOp{1}\).

\begin{proof}[Proof of Lemma \ref{lem:bo4}]
	For all \(N\in\N\), let \(q_N: \R \rightarrow \R\) be a function as specified by Lemma~\ref{lem:bo3}. Then,
	\(q_N\) is an elementary computable function. Moreover, as follows from \eqref{eq:fooI},
	\(q_N\) has a bounded \(L^1(\R)\)-norm. Consequently, we have \(q_N \in \CBs_{\pi}^{1}\) for all \(N\in\N\). 
	
	For all \(n\in\N\), set \(N(n) := 2\cdot 2^{96n + 96}\) and define the sequence \(\sq{f_n}_{n\in\N}\)
	according to
	\begin{align*}
		f_n := \frac{q_{N(n)}}{\lVert q_{N(n)} \rVert_{\Bs_{\pi}^{1}}}.
	\end{align*}
	Consequently, we have \(\lVert f_n\rVert_{\Bs_\pi^1} = 1\) for all \(n\in\N\). Moreover, 
	\(\sq{f_n}_{n\in\N}\) is a computable sequence of elementary computable functions.
	
	It remains to show that \(\lVert \SOp{1}f_n \rVert_{\ell^{1}} < \sfrac{1}{n}\) holds true for all \(n\in\N\). 
	From the definition of \(q_N\), it follows that for all \(N\in\N\), we have
	\begin{align}
		\lVert \SOp{1}q_N \rVert_{\ell^1} = 1 + \frac{1}{N}\sum_{k=1}^{N} 1 = 2.
	\end{align}
	Hence, for the sequence \(\sq{\SOp{1}f_n}_{n\in\N}\), we obtain
	\begin{align*}
		\lVert \SOp{1}f_n  \rVert_{\ell^1} 
			&< \frac{2}{\frac{1}{6 \pi} \ln \left( \frac{2\cdot 2^{96n + 96}}{2} \right) - \frac{1}{\pi}} \\
			&< \frac{2}{\frac{1}{6 \pi}\frac{1}{2} \log_2 \left( \frac{2\cdot 2^{96n + 96}}{2} \right) - \frac{1}{\pi}} \\
			&= \frac{2\pi}{\frac{1}{12} (96n + 96) - 1}	\\
			&= \frac{2\pi}{8n + 7}		\quad < \frac{1}{n},				
	\end{align*}
	which concludes the proof.
\end{proof}

\begin{lemma}\label{lem:bo5}
	Let $\Aset \subset \N$ be a recursively enumerable set. There exists a sequence
	$\sq{f_m^*}_{m \in \N}$ of elementary computable functions that satisfies the following:
	\begin{enumerate}
		\item The sequence \(\sq{x_m^*}_{m\in\N} = \sq{\SOp{1}f_m^*}_{m \in \N}\) is a computable sequence of sequences in $\ell^1$.
		\item We have \(\lVert f_m^* \rVert_{\Bs_{\pi}^{1}} = \mathds{1}_{A}(m)\) for all \(m\in\N\).
	\end{enumerate}
\end{lemma}\begin{proof}
	The derivation of the statement is analogous to the proof of Lemma \ref{lem:bo2}. 
	In order to avoid unnecessary redundancy, we restrict ourselves to a concise summary of the necessary steps.
	
	Let $\Aset \subset \N$ be a recursively enumerable set with runtime function
	\(g_\Aset : \N^2 \rightarrow \{0,1\}\). Consider the function \(h : \N^2 \rightarrow \N\)
	defined according to 
	\begin{align*}
		h(m,k) := \sum_{l=0}^{2^{k+2}} (1 - g_\Aset(m,l)).
	\end{align*}
	Further, let \(\sq{f_n}_{n\in\N}\)
	be a computable sequence of elementary computable functions as specified by Lemma \ref{lem:bo4}
	and define
	\begin{align*}
			f_{m,k} := f_{h(m,k)}
	\end{align*}
	for all \(m,k\in\N\). Then, \(\sq{f_{m,k}}_{m,k\in\N}\) is a computable double sequence of elementary computable functions.
	Moreover, the sequence \(\sq{x_{m,k}}_{m,k\in\N} := \sq{\SOp{1}f_{m,k}}_{m,k\in\N}\) is a computable double sequence of 
	elementary computable sequences. We define
	\begin{align*}
		f_m^* = 			\begin{cases}		\lim_{l\to\infty} 	f_{m,l} 	&\quad \text{if}~ m \in \Aset, \\
																	0			&\quad \text{otherwise},
							\end{cases}
	\end{align*}
	for all \(m\in\N\). Hence, \(\sq{f_m^*}_{m\in\N}\) is a sequence of elementary computable functions.
	Furthermore, the sequence \(\sq{x_m^*}_{m\in\N} = \sq{\SOp{1}f_m^*}_{m\in\N}\) is a sequence of elementary computable sequences
	and the sequence \(\sq{x_{m,k}}_{k\in\N}\) converges effectively towards \(x_m^*\) in \(\ell^1\),
	with respect to the recursive modulus of convergence \(\xi' : \N^2 \rightarrow \N,~ (m,K) \mapsto K\). As stated above,
	we refer to the proof of Lemma~\ref{lem:bo2} for details.
	
	It remains to show that the sequence \(\sq{f_m^*}_{m\in\N}\) satisfies \(\lVert f_m^*\rVert_{\Bs_{\pi}^{1}} = \mathds{1}_{A}(m)\) for all \(m\in\N\), which we again prove by case distinction.
	Recall that if \(m\in\Aset\) is satisfied, then \(f_m^*\) is an elementary computable function that such that \(f_n = f_m^*\) holds true for some \(n\in\N\). By assumption,
	we have \(\lVert f_n\rVert_{\Bs_{\pi}^{1}}  = 1\) for all \(n\in\N\). Hence, if \(m\in\Aset\) holds true, we have 
	\begin{align*}
			\lVert f_m^*\rVert_{\Bs_{\pi}^{1}} = \lVert f_n \rVert_{\Bs_{\pi}^{1}}  = 1.
	\end{align*}
	On the other hand, if \(m\in\Aset^{\complement}\) is satisfied, then \(f_m^*\) is the trivial elementary computable function, i.e., we have \(f_m^* = 0\).
	Thus, in this case, 
	\begin{align*}	\lVert f_m^*\rVert_{\Bs_{\pi}^{1}}  = 0
	\end{align*}
	holds true. We conclude that \(\sq{f_m^*}_{m\in\N}\) satisfies \(\lVert f_m^*\rVert_{\Bs_{\pi}^{1}} = \mathds{1}_{A}(m)\) for all \(m\in\N\),
	which concludes the proof.  	
\end{proof}

\begin{proof}[Proof of Theorem \ref{thm:UncompOne}]
	We prove the theorem by contradiction. As the line of reasoning is analogous to the proof of Theorem \ref{thm:UncompOne},
	we again restrict ourselves to a concise summary of the necessary steps.
	
	Let \(\Aset \subset \N\) be a recursively enumerable, nonrecursive subset of the natural numbers.
	Further, let \(\sq{f_m^*}_{m\in\N}\) be a sequence of elementary computable functions as specified by Lemma \ref{lem:bo5}. Define
	\(\sq{x_m^*}_{m\in\N} := \sq{\SOp{1}f_m^*}_{m\in\N}\) and consider a computable double sequence \(\sq{x_{m,k}}_{m,k\in\N}\) of elementary computable sequences
	such that for all \(m\in\N\), the sequence \(\sq{x_{m,k}}_{k\in\N}\) converges effectively towards \(x_m^*\) in \(\ell^1\),
	with respect to the recursive modulus of convergence \(\xi' : \N^2 \rightarrow \N,~ (m,K) \mapsto \xi'(m,K)\). 
	
	Assume now that 
	\(\TOp{1}\) is computable in the sense of Section \ref{sec:Conv}. Then,
	there exists a Turing machine \(\TMTOp{1}\) such that the pair
	\begin{align*}
		\mathfrak{F}^* = \big(\sq{f_{m,k}}_{m,k\in\N},\xi\big) := \TMTOp{1}\big(\sq{x_{m,k}}_{m,k\in\N}, \xi'\big)
	\end{align*}
	is a computable continuous-time description of the sequence \(\sq{f_m^*}_{m\in\N}\) in \(\CBs_{\pi}^{1}\). Further, we have 
	\begin{align*}
		\Big| \lVert f_m^*\rVert_{\Bs_\pi^1} - \lVert f_{m,k} \rVert_{\Bs_\pi^1} \Big| < \frac{1}{2^K}
	\end{align*}
	for all \(m,k,K\in\N\) that satisfy \(k \geq \xi(m,K)\), as follows from the triangle inequality. Moreover, the sequence \(\sq{ \lVert f_{m,k} \rVert_{\Bs_\pi^1}}_{m,k\in\N}\)
	is a computable double sequence of computable complex numbers, see Remark~\ref{rem:CompNorm}. Consequently, there exists computable double sequences \(\sq{r_{m,k}}_{m,k\in\N}\),
	\(\sq{s_{m,k}}_{m,k\in\N}\), of rational numbers, such that
	\begin{align*}
		\Big| \lVert f_m^* \rVert_{\Bs_\pi^1}  - (r_{m,k} +\iu s_{m,k}) \Big| < \frac{1}{2^K}
	\end{align*}
	holds true for all \(m,k,K\in\N\) that satisfy \(k \geq \xi(m,K)\). 
	
	Recall that by construction, we have \( \lVert f_m^* \rVert_{\Bs_\pi^1} = \mathds{1}_{\Aset}(m)\) for all \(m\in\N\). It follows that
	\begin{align*}
		m \in \Aset 	\quad \Leftrightarrow \quad 	\big(1 - r_{m,\xi(m,1)}\big)^2 + \big(s_{m,\xi(m,1)}\big)^2 < \frac{1}{4} 
	\end{align*}
	holds true for all \(m\in\N\). In particular, we have
	\begin{align*}
		\Aset = \Bigg\{m \in \N : \underbrace{\big(1 - r_{m,\xi(m,1)}\big)^2 + \big(s_{m,\xi(m,1)}\big)^2 < \frac{1}{4}}_{\equiv: P(m)} \Bigg\}.
	\end{align*}
	Since both \(\sq{r_{m,k}}_{m,k\in\N}\) and \(\sq{s_{m,k}}_{m,k\in\N}\) are computable sequences of rational numbers, the predicate \(P\)
	can be evaluated algorithmically, making \(\Aset\) a recursive set. The latter is a direct contradiction to the assumption, which concludes the proof.
\end{proof}

Furthermore, we obtain the following from the proof of Theorem~\ref{thm:UncompOne}:

\begin{corollary}\label{coro:BIBO}
	There does \emph{not} exist a Turing machine\linebreak \(\TM : \mathscr{X}^1 \rightarrow \stdR\) 
	such that 
	\begin{align}
		\big[\TM(\mathfrak{X})\big]_{\stdR} \equiv \lVert f \rVert_{\Bs_{\pi}^{1}} 
	\end{align}
	holds true for all \(f \in \CBs_{\pi}^{1}\), \(\mathfrak{X} \in \mathscr{X}^1\), that satisfy
	\(\SOp{1}f\equiv [\mathfrak{X}]_{\mathscr{X}}^1\).
\end{corollary}

In Section~\ref{sec:NandBsignals}, we introduced the BIBO-norm \(\lVert H \rVert_{\mathrm{BIBO}}\) of
an LTI system \(H\) with impulse response \(h\).
According to Corollary~\ref{coro:BIBO}, the BIBO-norm of \(H\) 
cannot always be computed based on a discrete-time description of \(h\).

\section{Computability of the Interpolation Operator for \MakeLowercase{\(1 < p < \infty\)}}\label{sec:CompOfT}
	So far, we have considered the limit cases of \(p = 1\) and \(p = \infty\) and shown the uncomputability
	of the associated interpolation operators by exploiting their discontinuity. For \(p\in\Rc\), \(1 < p < \infty\), the right hand side of the
	Plancherel--P\'olya inequality ensures the continuity of the interpolation operator \(\TOp{p} : \img\big(\SOp{p}\big) \rightarrow \Bs_{\pi}^{p}\),
	which we will make use of subsequently in order to show its computability.
	
	\begin{theorem}\label{thm:CompTNormal}
		The interpolation operator \(\TOp{p} : \img\big(\SOp{p}\big) \rightarrow \Bs_{\pi}^{p}\) is computable for \(p\in\Rc\), \(1 < p < \infty\).
	\end{theorem}
	
	The proof of Theorem \ref{thm:CompTNormal} follows the line of reasoning used in Section \ref{sec:Conv} 
	to derive the computability of the sampling operator \(\SOp{p}\) for \(1 < p < \infty\).

	\begin{proof}[Proof of Theorem \ref{thm:CompTNormal}]
		Let \(f\) be a function in \(\Cell_{\pi}^p\), \(p\in\Rc\), \(1 < p < \infty\),
		and assume that \(\mathfrak{X} = (\sq{x_n}_{n\in\N},\xi')\) is a computable discrete-time description thereof. That is, we have
		\begin{align*}
			\lVert \SOp{p}f - x_n\rVert_{\ell^{p}} < \frac{1}{2^M}
		\end{align*}
		for all \(n,M\in\N\) that satisfy \(n\geq \xi'(M)\).
		
		By definition, \(\sq{x_n}_{n\in\N}\) is characterized by a computable sequence \(\sq{\sq{c_{n,k}}_{k \in \mathcal{I}(n)}}_{n\in\N}\) of tuples of 
		computable complex numbers, where \(\mathcal{I}(n) = \{-L(n), \ldots, L(n)\}\) is a computable interval
		of natural numbers, that satisfy
		\begin{equation*}
		  x_n(m)
		  =
		  \sum_{k \in \mathcal{I}(n)} c_{n,k}\cdot \mathds{1}_{k}(m)
		\end{equation*}
		for all \(n\in\N, m\in\Z\). Hence, the sequence \(\sq{f_n}_{n\in\N} := \sq{\TOp{p}x_n}_{n\in\N}\) satisfies
		\begin{align*}
			f_n(z) = \sum_{k \in \mathcal{I}(n)} c_{n,k}\cdot \sinc(z - k)
		\end{align*}
		for all \(n\in\N\) and all \(z\in\C\), and is thus a computable sequence of elementary computable functions.
		
		In view of the Plancherel--P\'olya theorem, consider \(C_R\)
		with \(\Col{const:log-CR} := \lceil \log_2 C_R \rceil\). We have 
		\begin{align*}
		  2^{\Cor{const:log-CR}-M} 
				&\geq	C_R \lVert \SOp{p}f-x_n \rVert_{\ell^p} \\
				&=	C_R \lVert \SOp{p}(f - \TOp{p}x_n) \rVert_{\ell^p} \\
				&\geq \lVert f - \TOp{p}x_n \rVert_{\Bs_\pi^p} \\
				&= \lVert f - f_n \rVert_{\Bs_\pi^p}
		\end{align*}
		for all \(n,M \in \N\) that satisfy \(n \geq \xi'(M)\). 
		Define the function \(\xi :~ \N \rightarrow \N,~ M \mapsto \xi'(\Cor{const:log-CR} + M)\)
		and observe that \(\xi\) is recursive. We have
		\begin{align*}
			\lVert f - f_n \rVert_{\Bs_\pi^p} < 2^{-M}
		\end{align*}
		for all \(n,M\in\N\) that satisfy \(n\geq \xi(M)\). Hence, the sequence \(\sq{f_n}_{n\in\N}\) converges effectively
		towards \(f\in\CBs_\pi^p\), with respect to the recursive modulus of convergence \(\xi\). 
		In other words, \(\mathfrak{F} = (\sq{f_n}_{n\in\N},\xi)\)
		is a computable continuous-time description of \(f\).
	\end{proof}

\section{Interpretation in the Context of Digital Twins}\label{sec:InterpDT}
	In Section~\ref{sec:MotivDigitTwin}, we have already motivated the problem of converting continuous-time descriptions 
	into discrete-time descriptions and vice versa in the context of digital twin technology. Having established our results in 
	a formal manner, we want to come back in this section to the implications of our results for digital twin technology. For this purpose, 
	the case of \(p\in\{1,\infty\}\) is of particular interest, since it constitutes two relevant
	quantities for signal processing:
	\begin{itemize}
		\item In the space \(\Bs_{\pi}^{\infty}\), the mapping \(f \mapsto \lVert f\rVert_{\Bs_{\pi}^{\infty}}\) yields the
			peak value of the signal \(f\).
		\item In the space \(\Bs_{\pi}^{1}\), the mapping \(f \mapsto \lVert f\rVert_{\Bs_{\pi}^{1}}\) yields the
			BIBO-norm of LTI systems with impulse response \(f\).
	\end{itemize}
	In both spaces, we considered two different types of machine-readable signal descriptions:
	\begin{itemize}
		\item A continuous-time description \(\mathfrak{F} = (\sq{f_n}_{n\in\N},\xi) \in \mathscr{F}^p\) of some signal \(f\in\CBs_{\pi}^{p}, p\in\{1,\infty\}\),
			consisting of a computable sequence \(\sq{f_n}_{n\in\N}\) of elementary computable signals in \(\CBs_{\pi}^{p}\) and a recursive modulus of convergence
			\(\xi:\N\rightarrow\N\) such that \(\sq{f_n}_{n\in\N}\) converges effectively towards \(f\) in \(\CBs_{\pi}^{p}\) with respect to \(\xi\).
		\item A discrete-time description \(\mathfrak{X} = (\sq{x_n}_{n\in\N},\xi') \in \mathscr{X}^p\) of some signal \(f\in\CBs_{\pi}^{p}, p\in\{1,\infty\}\),
			consisting of a computable sequence \(\sq{x_n}_{n\in\N}\) of elementary computable sequences in \(\Cell^{p}\) and a recursive modulus of convergence
			\(\xi':\N\rightarrow\N\) such that \(\sq{x_n}_{n\in\N}\) converges effectively towards \(\sq{f(k)}_{k\in\Z}\) in \(\Cell^{p}\) with respect to \(\xi'\).
	\end{itemize}
	Both \(\mathfrak{F}\) and \(\mathfrak{X}\) are digital twins of the signal \(f\). According to a generalization of the Plancherel--P\'olya theorem 
	(c.f. Section~\ref{sec:NandBsignals}), \(\mathfrak{F}\) and \(\mathfrak{X}\) are analytically equivalent, in the sense that they characterize \(f\) uniquely. 
	However, we have seen that \(\mathfrak{F}\) and \(\mathfrak{X}\) are not algorithmically equivalent: given the discrete-time description \(\mathfrak{X}\),
	it is generally not possible to compute a continuous-time description \(\mathfrak{F}\) of the same signal.
	
	In signal processing and information theory, the discrete-time description \(\mathscr{X}^p\) is the standard way of uniquely characterizing signals \(f\in\CBs_{\pi}^p\). For \(f\in\CBs_{\pi}^{1}\) fixed,
	the number \(\lVert f \rVert_{\Bs_{\pi}^{1}}\) is computable. Hence, there always exists an algorithm that computes the number  
	\(\lVert f \rVert_{\Bs_{\pi}^{1}}\) as such. However, there does \emph{not} exist an algorithm that computes \(\lVert f \rVert_{\Bs_{\pi}^{1}}\)
	given a discrete-time description \(\mathfrak{X}\in\mathscr{X}^1\) of \(f\). In other words, \(\mathfrak{X}\in\mathscr{X}^1\) is not a ``feasible input'' for computing 
	\(\lVert f \rVert_{\Bs_{\pi}^{1}}\). The same holds true for the space \(\CBs_{\pi}^{\infty}\) and the associated norm.
	For a continuous-time description \(\mathscr{F}^p\), on the other hand, this is always possible, which shows again that both descriptions are not 
	algorithmically equivalent. This finding is especially important in view of the fact that \(\mathscr{X}^p\) is the standard description of 
	signals in digital signal processing.
	
	In general, digital twins are, by very nature, the only way to encode the physical world into a machine-readable manner. 
	In this context, it is necessary for the digital twin to be able to represent the essential properties of the physical system.
	The digital twin hence describes the suitable description of the analog world as an input for the digital computer. 
	
\section{Conclusion}\label{sec:concl}

Within the scope of this article, we have, in the context of digital twinning, 
considered the task of converting different digital descriptions of analog, bandlimited signals into 
each other. We have shown that the computability of quantities associated to the real world, analog system
based on its digital counterpart depends crucially on the choice of the \emph{proper} language to describe 
the analog world. Furthermore, we have shown that different languages (and hence different digital twins) are not
necessarily algorithmically equivalent, i.e., they cannot be converted into each other in an algorithmic manner.
		
To the best of the authors' knowledge, this work is the first to consider the problem of
identifying ``suitable'' digital twin descriptions of analog systems. Even though two different machine-readable
languages may equally be able to uniquely describe the elements of a set of real-world objects,
they may still not be equally powerful in an algorithmic sense. That is, not all information about the real-world object
may be algorithmically ``accessible'' from a description of that object in a certain language. Therefore, it is necessary to 
select the suitable language for each implementation of a digital twin carefully.
This is particularly crucial for future applications of cyber-physical systems in healthcare and robotics.
		
Modern applications of digital twin technologies require the strict adherence to specifications of technological \emph{trustworthiness},
a collective term for the principles of privacy, secrecy, safety, resilience, availability, accountability, authenticity, device independence, 
reliability and integrity~\cite{FeBo22}. In~\cite{BoBoMo22}, we have already hinted towards the possibility that different types 
of digital descriptions of analog signals may play a role in the context of digital twinning in metaverse applications, a technology 
that is to be standardized in upcoming 5G releases. In this work, we presented theoretical findings on how the algorithmic non-equivalence of such 
descriptions may indeed compromise a digital twin's trustworthiness. In particular, our results affect reliability and integrity. The latter refers to
a technological system's ability to correctly and reliably operate within a specified margin of service, including (in the context of digital twins)
the correct recording of the state of the system's physical components, and the ability to detect faulty modes of operation. We have seen
that a digital twin of a bandlimited system in the language \(\mathscr{F}^p\) can provide this type of integrity, while it is violated
with respect to the language \(\mathscr{X}^p\): A digital twin of a bandlimited system in the language \(\mathscr{X}^p\)  
cannot properly capture the relevant parameters of its analog counterpart in all cases. Consequently, for general applications of digital twin technology,
it is first necessary to prove that the language chosen to implement a digital twin is in theory able to solve the relevant task. 
		
One of the main results of digital computability theory is the proof of the existence of universal machines. 
That is, for every digital computer and every universal machine, there exists a compiler that translates 
every program for the digital computer into an equivalent program for the universal machine. 
Thus, regarding input and output behavior, a universal machine can simulate any other digital computer.
For computing on bit strings, the interpretation of universal machines is imminent, 
since all possible computers operate on the same set. Conversely, for more complex problems, 
the set of admissible machines must first be restricted. In our case, we can interpret
the spaces \(\CBs_{\pi}^{p}\) and \(\Cell_{\pi}^{p}\) for \(p\in\Rc\), \(1 \leq p \leq \infty\), as abstract sets.
Then, a Turing machine that computes the operators \(\SOp{p}\) or \(\TOp{p}\) in the sense of Section~\ref{sec:Conv}
may be regarded as a compiler.  For \(p\in\Rc\), \(1 \leq p \leq \infty\), each sequence of sampling values \(\sq{f(k)}_{k\in\N} \in \Cell^{p}\)
uniquely characterizes a signal \(f\in\CBs_{\pi}^{p}\), but a corresponding continuous-time description
cannot be determined algorithmically in the case of \(p\in\{1,\infty\}\), i.e. compiler problem is not solvable. 
The transformation of descriptions in different languages into one another is therefore generally a ``creative'' process that cannot be automated.  
Yet, as stated above, the ability to automatically detect faulty modes of operation can be an essential part of the integrity
requirements placed upon a technological system. For details, we again refer to~\cite{FeBo22}. In order to meet this requirement, the choice of the proper 
machine-readable description of the physical system components is crucial.

\bibliographystyle{elsarticle-num-names}
\bibliography{noauthorhyphen,IEEEfull,names_long,Compilerproblem,sampling_transit_export}

\end{document}